\documentclass[conference]{IEEEtran}
\IEEEoverridecommandlockouts
\usepackage{cite}
\usepackage{amsmath,amssymb,amsfonts}
\usepackage{amsthm}
\usepackage{mathrsfs}
\usepackage[ruled,linesnumbered]{algorithm2e}
\usepackage{graphicx}
\usepackage{textcomp}
\usepackage{xcolor}
\usepackage{multirow}
\usepackage[normalem]{ulem}
\usepackage{booktabs}
\usepackage{setspace}
\usepackage{soul}
\usepackage{color, xcolor}

\graphicspath{{./photo/}}
\soulregister{\cite}7
\soulregister{\ref}7

\def\BibTeX{{\rm B\kern-.05em{\sc i\kern-.025em b}\kern-.08em
    T\kern-.1667em\lower.7ex\hbox{E}\kern-.125emX}}

\begin{document}
\title{Towards Semantic Consistency: Dirichlet Energy Driven Robust Multi-Modal Entity Alignment
}

\author{
\IEEEauthorblockN{
Yuanyi Wang\IEEEauthorrefmark{2}$^{\#}$,
Haifeng Sun\IEEEauthorrefmark{2}$^{\#}$ \thanks{\# Equal contribution. },
Jiabo Wang\IEEEauthorrefmark{2},
Jingyu Wang\IEEEauthorrefmark{2}$^{\ast}$ \thanks{* Corresponding author},
Wei Tang\IEEEauthorrefmark{2},
Qi Qi\IEEEauthorrefmark{2},
Shaoling Sun\IEEEauthorrefmark{3}, and
Jianxin Liao\IEEEauthorrefmark{2}
}
\IEEEauthorblockA{\IEEEauthorrefmark{2}State Key Laboratory of Networking and Switching Technology, \\ Beijing University of Posts and Telecommunications, Bejing, China}
\IEEEauthorblockA{\IEEEauthorrefmark{3}China Mobile (Suzhou) Software Technology Co., Ltd., Jiangsu, China}
\IEEEauthorblockA{\{wangyuanyi,hfsun,jiabowang,wangjingyu,tangocean,qiqi8266\}@bupt.edu.cn, \\sunshaoling@cmss.chinamobile.com, jxlbupt@gmail.com}}
\IEEEpeerreviewmaketitle
\maketitle

\begin{abstract}
Multi-Modal Entity Alignment (MMEA) is a pivotal task in Multi-Modal Knowledge Graphs (MMKGs), seeking to identify identical entities by leveraging associated modal attributes. However, real-world MMKGs confront the challenges of semantic inconsistency arising from diverse and incomplete data sources. This inconsistency is predominantly caused by the absence of specific modal attributes, manifesting in two distinct forms: disparities in attribute counts or the absence of certain modalities. Current methods address these issues through attribute interpolation, but their reliance on predefined distributions introduces modality noise, compromising original semantic information. Furthermore, the absence of a generalizable theoretical principle hampers progress towards achieving semantic consistency. In this work, we propose a generalizable theoretical principle by examining semantic consistency from the perspective of Dirichlet energy. Our research reveals that, in the presence of semantic inconsistency, models tend to overfit to modality noise, leading to over-smoothing and performance oscillations or declines, particularly in scenarios with a high rate of missing modality. To overcome these challenges, we propose DESAlign, a robust method addressing the over-smoothing caused by semantic inconsistency and interpolating missing semantics using existing modalities. Specifically, we devise a training strategy for multi-modal knowledge graph learning based on our proposed principle. Then, we introduce a propagation strategy that utilizes existing features to provide interpolation solutions for missing semantic features. DESAlign outperforms existing approaches across 60 benchmark splits, encompassing both monolingual and bilingual scenarios, achieving state-of-the-art performance. Experiments on splits with high missing modal attributes demonstrate its effectiveness, providing a robust MMEA solution to semantic inconsistency in real-world MMKGs.
\end{abstract}

\begin{IEEEkeywords}
Knowledge graph, Multimodal, Entity alignment, Semantic Learning, Semantic consistency, Dirichlet energy
\end{IEEEkeywords}

\section{Introduction}
\label{introduction}

Knowledge graphs (KGs) have emerged as a prominent data structure for representing factual knowledge through RDF triples. Recently, the expansion of KGs has led to the development of Multi-Modal Knowledge Graphs (MMKGs), incorporating diverse information types such as text, images, videos, and more. MMKGs serve as a foundation for cross-modal tasks including information retrieval\cite{b1,b2,b5} and question answering\cite{b3,b4}. Despite their growing significance, MMKGs are hindered by limitations in coverage, restricting their effectiveness in downstream applications. Addressing the integration of heterogeneous MMKGs, Multi-Modal Entity Alignment (MMEA) becomes a pivotal challenge. MMEA aims to identify equivalent entity pairs across different MMKGs by integrating the attribute information across modalities.

MMEA is usually treated as a supervised problem\cite{chen2020mmea,m2,mclea,m4,m5}, relying on extensive seed alignments as supervised signals. These signals guide the learning process of entity embeddings, enabling the projection of MMKGs with varying modal attributes into a unified semantic embedding space. Subsequently, alignment results are predicted by calculating the pairwise similarity between these semantic embeddings. However, MMKGs are typically assembled from diverse and incomplete data sources, leading to semantic inconsistency. This inconsistency arises from the absence of specific modal attributes, leading to disparities in attribute counts or the absence of certain modalities. Illustrated in Figure \ref{Sinconsistency}, the semantics derived from diverse modal attributes will exhibit differences, resulting in significant alignment errors. Specifically, the entity \textit{Elon Musk} in $G_s$ is associated with only one text attribute \textit{Canada}. Conversely, the corresponding entity \textit{Elon Reeve Musk} in $G_t$ exhibits richer attributes, including textual information and images. This discrepancy in attribute counts and modality richness introduces alignment pairs (\textit{Elon Musk} and \textit{Elon Reeve Musk}) with differing semantics. This incompleteness is inherent in real-world KGs like DBpedia\cite{lehmann2015dbpedia} and Freebase\cite{bollacker2008freebase}, where not all entities are associated with modal attributes. For example, in common bilingual datasets such as DBP15K$_{JA-EN}$\cite{m6} and monolingual datasets like FBYG15K\cite{liu2019mmkg}, only 67.58\% and 73.24\% of entities have associated images, respectively.

\begin{figure}[t]
\centering
\includegraphics[width = \linewidth]{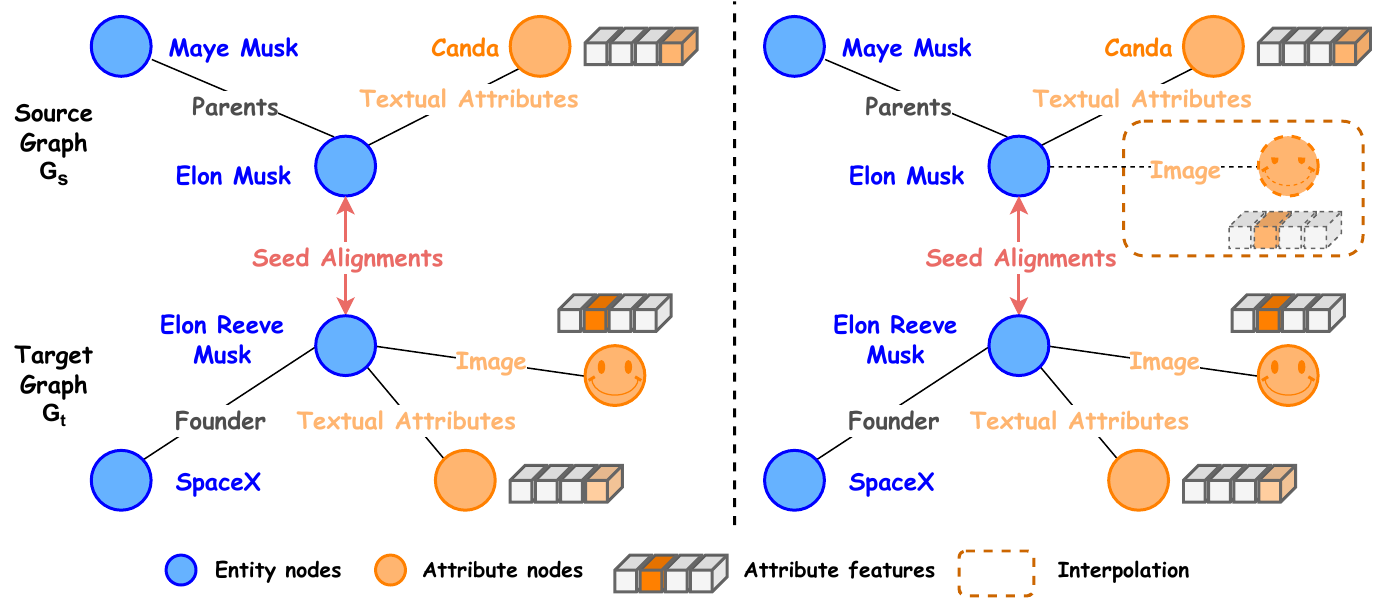}
\caption{An example of semantic inconsistency issue and interpolation process in the MMEA task between MMKG1 and MMKG2.}
\label{Sinconsistency}
\end{figure}

Existing MMEA methods address semantic inconsistency through various strategies. MCLEA\cite{mclea} and MEAformer\cite{meaformer} use predefined distributions to interpolate absent modal attributes and aim to tackle over-smoothing, introducing modal noise. ACK\cite{ack} unifies entities with a common attribute but may lose valuable semantic information. UMAEA\cite{UMAEA} uses a multi-stage approach focused on visual modality. However, none of these methods provide an end-to-end solution. Furthermore, these methods predominantly rely on predefined rules that concentrate on a single modality, ignoring that the semantic inconsistency issue is prevalent across all modalities.

Beyond semantic inconsistency, another challenge is the lack of a generalizable theoretical principle guiding multi-modal knowledge graph learning with semantic consistency. Existing approaches often rely on high-quality predefined distributions to interpolate missing attributes, yielding comparable or inferior results, especially with the high rate of missing modal attributes. Recently, the Dirichlet energy metric\cite{maskey2023fractional} has been introduced to promote homophily within graphs, smoothing graph embeddings. The key concept involves the gradient of Dirichlet energy approaching zero when nodes preserve homophily, driving embeddings to converge. We identify Dirichlet energy as a potential avenue for understanding and guiding semantic consistency in MMEA.

Therefore, our work addresses three critical challenges in achieving semantic consistency within MMEA. Firstly, there is a lack of a generalizable theoretical principle and comprehensive analysis. Second, an effective method for addressing over-smoothing caused by semantic inconsistency across all modalities is needed. Current methodologies often target individual modalities and perform poorly in the high rate of missing modal attributes. Finally, even with an effective framework addressing over-smoothing, the root cause of semantic inconsistency, the lack of certain modal attributes, remains unsolved, leaving essential semantic information missing.

To address these challenges, this research introduces a comprehensive and theoretically grounded framework, DESAlign (\textbf{\underline{D}}irichlet \textbf{\underline{E}}nergy driven \textbf{\underline{S}}emantic-consistent multi-modal entity \textbf{\underline{Align}}ment). Our approach analyzes the issue of semantic inconsistency from the perspective of Dirichlet energy, drawing connections between graph Dirichlet energy and semantic interpolation quality. Then we propose a generalizable theoretical principle inspired by Dirichlet energy to guide multi-modal knowledge graph learning, forming the basis for our \textbf{M}ulti-\textbf{M}odal \textbf{S}emantic \textbf{L}earning (MMSL). This learning strategy is specifically tailored to counter the over-smoothing issue caused by semantic inconsistency. Furthermore, we introduce a propagation strategy for interpolation, \textbf{S}emantic \textbf{P}ropagation (SP), to address the challenge of missing semantics resulting from absent modal attributes. Specifically, we demonstrate that interpolation of missing modal features can be achieved through existing features, employing a discretization scheme to derive \textit{explicit Euler solutions}. We also delve into the theoretical foundations of DESAlign from the perspective of Dirichlet energy, offering valuable insights into its operational mechanisms and the underlying principles of semantic consistency. We propose a new set of MMEA datasets with a total of 60 splits. Our method demonstrates superior performance when compared to state-of-the-art MMEA methods across proposed datasets, including scenarios that encompass both monolingual and bilingual settings, in both weakly supervised and normal supervised conditions. Notably, DESAlign excels not only in basic configurations but also when employed with an iterative strategy, making it highly versatile.

To sum up, our contributions are as follows:
\begin{itemize}
\item \textbf{Generalizable Principle:} We propose a generalizable theoretical principle motivated by Dirichlet energy to guide multi-modal knowledge graph learning, ensuring semantic consistency optimization. 
\item \textbf{Interpolation Solutions:} We propose a semantic propagation strategy demonstrating that missing modal features can be interpolated by existing modalities. We introduce a discretization scheme, deriving explicit Euler solutions, paving the way for continuous methods promoting semantic consistency on graphs.
\item \textbf{Innovative Framework:} The innovative framework, DESAlign, is proposed based on our Dirichlet energy-driven principles. Specifically designed for MMEA, DESAlign addresses over-smoothing issues in the training process and interpolates missing semantic information.
\item \textbf{Extensive Experiments:} Conducting comprehensive experiments on real-world datasets across monolingual and bilingual contexts, under both normal and weakly supervised settings. DESAlign consistently achieves state-of-the-art performance. Explorations on extensive benchmarks with the high ratio of missing modal attributes establish DESAlign's superiority over existing methods, highlighting its robustness against challenges associated with missing modal attributes and diverse modalities. 
\footnote{The code and data are available at https://github.com/wyy-code/DESAlign.}
\end{itemize}

The paper is organized as follows: Section \ref{Preliminary} covers preliminaries of MMEA and Dirichlet energy. Section \ref{Dirichletanalysis} discusses semantic inconsistency in MMEA and introduces our principle. Section \ref{Methodology} details our DESAlign framework. Experimental findings are reported in Section \ref{MAINExperiments}. Sections \ref{relatedwork} and \ref{concluision} review related work and conclude the study, respectively.

\section{Preliminary}
\label{Preliminary}

Let $G = (\mathcal{E},\mathcal{R},\mathcal{A}, \mathcal{V})$ be a MMKG with an adjacency matrix $\mathbf{A}$. Here, $\mathcal{E}$, $\mathcal{R}$, $\mathcal{A}$, and $\mathcal{V}$ correspond to entities, relations, textual attributes, and images, respectively. Assuming $G$ as an undirected graph, the graph Laplacian matrix, denoted as $\mathbf{\Delta}$, is positive semi-definite, given by $\mathbf{\Delta} = \mathbf{I} - \mathbf{\widetilde{A}}$, where $\mathbf{\widetilde{A}} = \mathbf{D}^{-\frac{1}{2}}\mathbf{A}\mathbf{D}^{-\frac{1}{2}}$ is the normalized adjacency matrix, and $\mathbf{D} = \text{diag}(\sum_ja_{1j},\dots,\sum_ja_{nj})$ is the diagonal degree matrix.

\subsection{Problem Definition}
This paper addresses the challenge of aligning two MMKGs with semantic inconsistency issues, denoted as source ($G_s$) and target ($G_t$) graphs. MMEA aims to establish a one-to-one mapping $\Phi = \{(e_i, e_i^\prime) | e_i \in \mathcal{E}_s, e_i^\prime \in \mathcal{E}_t, e_i \equiv e_i^\prime\}$ using seed alignments $\Phi^\prime \in \Phi$, where $\equiv$ represents an equivalence relation between $e_i$ and $e_i^\prime$.

\theoremstyle{plain}
\newtheorem{myDef}{\noindent\bf Definition}
\begin{myDef}
(Multi-Modal Entity Alignment) MMEA seeks a one-to-one mapping $\Phi$ from the source graph $G_s = (\mathcal{E}_s, \mathcal{R}_s, \mathcal{A}_s, \mathcal{V}_s)$ to the target graph $G_t = (\mathcal{E}_t, \mathcal{R}_t, \mathcal{A}_t, \mathcal{V}_t)$, relying on seed alignments $\Phi^\prime \in \Phi$. This mapping enables the learning of entity representations $\mathbf{X}$ in the unified semantic space for pairwise similarity assessment.
\end{myDef}
The semantic embeddings evolve through a multi-layer semantic encoder, with the transformation from initial features $\mathbf{X}^{(0)}$ to final entity semantic embedding $\mathbf{X}^{(k)}$ described by:
\begin{equation}
    \mathbf{X}^{(k)}=\sigma (W^{(k)}\mathbf{X}^{(k-1)})
\end{equation}
where $\mathbf{X}$ is a $d$-dimensional embedding, $W^{(k)} \in \mathbb{R}^{d \times d}$ denotes trainable weights for feature transformation, and $\sigma$ represents an activation function.

\subsection{Semantic Consistency by Interpolation}
For simplicity, let $\mathbf{x} \in \mathbb{R}^{n}$ represent scalar features, with $n$ denoting the number of entities in graphs $G_s$ and $G_t$. As discussed in Section \ref{introduction}, we define the set of entities exhibiting \textit{semantic consistent} as $\mathcal{E}_c \subseteq \mathcal{E}$, while the remaining \textit{inconsistent} entities are split into $\mathcal{E}_o = \mathcal{E} \backslash \mathcal{E}_c$. These inconsistent entities encompass two subsets: $\mathcal{E}_{o1}$ with differing attribute counts, and $\mathcal{E}_{o2}$ with missing modal attributes. Assuming entity ordering, we represent the semantic feature matrix as:
\begin{equation}
\label{SCA}
\mathbf{x} = 
\left [ 
    \begin{array}{lc}
        \mathbf{x}_c \\
        \mathbf{x}_{o1}\\
        \mathbf{x}_{o2}\\
    \end{array}
\right ],
\mathbf{A} = 
\left [ 
    \begin{array}{lcc}
        \mathbf{A}_{cc} & \mathbf{A}_{co1} & \mathbf{A}_{co2}\\
        \mathbf{A}_{o1c} & \mathbf{A}_{o1o1} & \mathbf{A}_{o1o2}\\
        \mathbf{A}_{o2c} & \mathbf{A}_{o2o1} & \mathbf{A}_{o2o2}\\
    \end{array}
\right ]
\end{equation}
Given the undirected nature of the graph, we have $\mathbf{A}$ symmetric, implying $\mathbf{A}^\top_{co1} = \mathbf{A}_{o1c}$, $\mathbf{A}^\top_{co2} = \mathbf{A}_{o2c}$, and $\mathbf{A}^\top_{o1o2} = \mathbf{A}_{o2o1}$. The Laplacian matrix $\mathbf{\Delta}$ shares these properties. The subsequent discussion assumes these symmetries. We now present the definition of feature interpolation.
\begin{myDef}
(Feature Interpolation) Given the semantic features $\mathbf{X}$ and graph structure $\mathbf{A}$, the interpolation function $f$ modifies inconsistent semantic features as $\mathbf{\widehat{X}} = f(\mathbf{X}, \mathbf{A})$.
\end{myDef}
Unlike existing methods that rely on predefined rules to interpolate initial features $\mathbf{X}^{(0)}$, our approach innovatively interpolates final semantic features $\mathbf{X}^{(k)}$ through gradient flow-induced information propagation within the graph \cite{wang2024gradient}.

\subsection{Dirichlet Energy for Graphs}
Graph Dirichlet energy (DE), measuring node smoothness within a graph \cite{zhou2021dirichlet}, acts as a pivotal energy function that encapsulates prior knowledge of graph feature behavior, particularly in imputation tasks \cite{zhang2023missing}. This function, premised on graph homophily, employs DE to quantify feature smoothness. 

\begin{myDef}
    (Dirichlet Energy) Given the entity embedding $\mathbf{X}\in \mathbb{R}^{N\times d}$, the corresponding Dirichlet energy is defined by:
    \begin{equation}
        \mathscr{L}(\mathbf{X})=tr(\mathbf{X}^\top \Delta \mathbf{X}) = \frac{1}{2}\sum_{i,j=1}^Na_{i,j}||\frac{\mathbf{X}_i}{\sqrt{\mathbf{D}_{i,i}+1}} - \frac{\mathbf{X}_j}{\sqrt{\mathbf{D}_{j,j}+1}}||^2_2
    \end{equation}
\end{myDef}
Dirichlet Energy, with its solid theoretical foundation and spectral properties \cite{spielman2012spectral}, facilitates an in-depth embedding distribution analysis by leveraging its mathematically local characteristics and constraints (like Lipschitz conditions), distinguishing it from other metrics such as H@k and MRR.

\section{Dirichlet Energy for Semantic Learning}
\label{Dirichletanalysis}
In this section, we delve into the profound connection between semantics and Dirichlet energy within the context of MMEA. We illustrate that semantic inconsistency, viewed through the perspective of Dirichlet energy, causes the over-smoothing issue. Subsequently, we introduce a generalizable theoretical principle for semantic learning.

\subsection{Semantic Quality}
Semantic inconsistency, a common hurdle in MMEA, occurs due to absent modal attributes, complicating the assessment of multi-modal knowledge graph representation and interpolation methods. To tackle this, we advocate for evaluating semantic quality through Dirichlet energy, establishing its relevance to semantic embedding's integrity. The subsequent proposition delineates the relationship between Dirichlet energy and semantic embedding quality:
\theoremstyle{plain}
\newtheorem{myTh}{\bf Proposition}
\begin{myTh}
Let $\langle \cdot \rangle$ denote the inner product. Dirichlet energy can reflect the local quality of semantic interpolation between the modified feature matrix $\mathbf{\widehat{X}}$ and the original $\mathbf{X}$ through the following equation:
    \begin{equation}
        \mathscr{L}(\mathbf{\widehat{X}}) - \mathscr{L}(\mathbf{X}) \ge 2\langle \mathbf{\Delta}\mathbf{X}, (\mathbf{\widehat{X}}-\mathbf{X}) \rangle
    \end{equation}
\end{myTh}
\begin{proof}
    As the graph Laplacian matrix is positive semi-definite, the Dirichlet energy $\mathscr{L}$ is convex. According to the properties of convex functions and the \textit{Taylor equation}, we have:
    \begin{spacing}{0.6}
     \begin{equation*}
     \begin{aligned}
         \mathscr{L}&(\mathbf{\widehat{X}}) \ge  \mathscr{L}(\mathbf{X}) + \langle \nabla \mathscr{L}(\mathbf{X}), (\mathbf{\widehat{X}}-\mathbf{X}) \rangle \\
         \mathscr{L}&(\mathbf{\widehat{X}}) - \mathscr{L}(\mathbf{X}) \ge 2\langle \mathbf{\Delta}\mathbf{X}, (\mathbf{\widehat{X}}-\mathbf{X}) \rangle    
     \end{aligned}
     \end{equation*}
     \end{spacing}
\end{proof}
This establishes Dirichlet energy as a metric for local semantic interpolation quality. The operator $\mathbf{\Delta}\mathbf{X}$ acts as a graph filter, with the inner product providing a measure of how well the enhanced features in $\mathbf{\widehat{X}}$ conform to the graph inherent structure, as dictated by $\mathbf{\Delta}\mathbf{X}$. Thus, a higher inner product implies that $\mathbf{\widehat{X}}$ modifications are in harmony with the spectral properties of the graph, ensuring that interpolations preserve and reflect the graph's semantic nuances effectively.
Based on \cite{zhang2023missing} and our analysis, Dirichlet energy can be utilized to simultaneously constrain the upper and lower bounds on the quality of interpolation. Note that the eigenvalues of $\mathbf{\Delta}$ vary with real-world graphs and lie within the range $[0, 2)$ \cite{chung1997spectral}.
\theoremstyle{plain}
\newtheorem{myCo}{\bf Corollary}
\begin{myCo}
\label{corollary1}
    Let $\lambda_{max}$ be the largest eigenvalue of $\mathbf{\Delta}$, $\bold M$ be the maximum of $l_2$ norms $||\mathbf{\widehat{X}}||_2, ||\mathbf{X}||_2$, and $\bold m$ be the minimum. The optimal interpolation quality is limited by the Dirichlet energy gap as follows:
    \begin{equation}
    \label{Corollary}
        \frac{|\mathscr{L}(\mathbf{\widehat{X}}) - \mathscr{L}(\mathbf{X})|}{2\lambda_{max}\bold m} \ge ||\mathbf{\widehat{X}} - \mathbf{X}||_2 \ge \frac{|\mathscr{L}(\mathbf{\widehat{X}}) - \mathscr{L}(\mathbf{X})|}{2\lambda_{max}\bold M}
    \end{equation}    
\end{myCo}
\begin{proof}
    Let $\mathbf{X}_a$ and $\mathbf{X}_b$ be two bounded feature matrices, and $\bold M, \bold m$ be the $l_2$ boundary. Considering $\mathscr{L}(\mathbf{X})$ as a local Lipschitz function when $\mathbf{X}$ is bounded, we have:
    \begin{spacing}{0.2}
    \begin{equation*}
    \begin{aligned}
        \frac{|\mathscr{L}(\mathbf{X}_a) - \mathscr{L}(\mathbf{X}_b)|}{||\mathbf{X}_a - \mathbf{X}_b||_2} &\le \sup_{||\mathbf{X}||_2 \le \bold M}|| \nabla \mathscr{L}(\mathbf{X})^\top ||_2\\
        & =\sup_{||\mathbf{X}||_2 \le \bold M}|| 2\mathbf{\Delta}\mathbf{X} ||_2 \le 2\lambda_{max} \bold M \\
        \frac{|\mathscr{L}(\mathbf{X}_a) - \mathscr{L}(\mathbf{X}_b)|}{||\mathbf{X}_a - \mathbf{X}_b||_2} &\ge \inf_{||\mathbf{X}||_2 \ge \bold m}|| \nabla \mathscr{L}(\mathbf{X})^\top ||_2\\
        & =\inf_{||\mathbf{X}||_2 \ge \bold m}|| 2\mathbf{\Delta}\mathbf{X} ||_2 \ge 2\lambda_{max} \bold m 
    \end{aligned}        
    \end{equation*}
    \end{spacing}
\end{proof}
Existing metrics such as C@k \cite{huang2022semantic} tackle structural inconsistencies, and downstream metrics like MRR assess semantic variations. However, they fall short of directly evaluating semantic interpolation. Our work bridges this gap by linking Dirichlet energy to semantic quality, defining clear bounds for interpolation (Corollary \ref{corollary1}). 

Additionally, our analysis from the Dirichlet energy perspective innovates an efficient interpolation method detailed in Section \ref{Semantic Propagation}. We elucidate the interpolation process through Dirichlet energy's gradient flow, with its discretization providing a straightforward yet potent approach.

\subsection{Semantic Consistent Learning}
Dirichlet energy plays a crucial role in learning semantic representations. We introduce a robust principle to enhance the ability of semantic encoder to capture high-level semantic information, as detailed in Proposition \ref{pro2}. This principle sets lower and upper bounds for semantic representation quality. The semantic learning process is theoretically
simplified by excluding the non-linear activation function $\sigma$, representing the semantic encoder as $\mathbf{X}^{(k)}=W^{(k)}W^{(k-1)}\dots W^{(1)}\mathbf{X}$.
\begin{myTh}
\label{pro2}
    The Dirichlet energy of the semantic feature at the $k$-th layer is bounded as follows:
    \begin{equation}
        p_{min}^{(k)}\mathscr{L}(\mathbf{X}^{(k-1)}) \le \mathscr{L}(\mathbf{X}^{(k)}) \le p_{max}^{(k)}\mathscr{L}(\mathbf{X}^{(k-1)})
    \end{equation}
    where $p_{min}^{(k)}$ and $p_{max}^{(k)}$ are the squares of minimum and maximum singular values of weight $W^{(k)}$, respectively.
\end{myTh}
\begin{proof}
   The upper bound is derived as follows:
    \begin{equation*}
    \begin{aligned}
        \mathscr{L}(\mathbf{X}^{(k)}) & = \mathscr{L}(W^{(k)}\mathbf{X}^{(k-1)}) \\ & = tr((W^{(k)}\mathbf{X}^{(k-1)})^\top\mathbf{\Delta}(W^{(k)}\mathbf{X}^{(k-1)})) \\ & = tr((\mathbf{X}^{(k-1)})^\top\mathbf{\Delta}\mathbf{X}^{(k-1)}(W^{(k)})^\top W^{(k)}) \\
        & \le tr((\mathbf{X}^{(k-1)})^\top\mathbf{\Delta}\mathbf{X}^{(k-1)})\psi_{max}((W^{(k)})^\top W^{(k)}) \\ & = \mathscr{L}(\mathbf{X}^{(k-1)})p_{max}^{(k)}
    \end{aligned}        
    \end{equation*}
    $\psi_{max}$ denotes the maximum eigenvalue of a matrix. Since $ tr((\mathbf{X}^{(k-1)})^\top\mathbf{\Delta}\mathbf{X}^{(k-1)}) \ge 0$, we can obtain the above inequality relationship. Similarly, the lower bound is derived as:
    \begin{spacing}{0.3}
     \begin{equation*}
    \begin{aligned}
        \mathscr{L}(\mathbf{X}^{(k)}) & = tr((\mathbf{X}^{(k-1)})^\top\mathbf{\Delta}\mathbf{X}^{(k-1)}(W^{(k)})^\top W^{(k)}) \\
        & \ge tr((\mathbf{X}^{(k-1)})^\top\mathbf{\Delta}\mathbf{X}^{(k-1)})\sigma_{min}((W^{(k)})^\top W^{(k)}) \\ & = \mathscr{L}(\mathbf{X}^{(k-1)})p_{min}^{(k)}
    \end{aligned}        
    \qedhere
    \end{equation*}
    \end{spacing}
\end{proof}
This proposition asserts that without proper design and training on the weight $W^{(k)}$, the Dirichlet energy can become either too small or too large. Based on the common Glorot initialization \cite{glorot2010understanding} and $l_2$ regularization, we empirically observe that some weight matrices of semantically inconsistent features tend to approximate zero in higher feedforward layers. Consequently, the corresponding square singular values are close to zero in these intermediate layers. This phenomenon results in the Dirichlet energy becoming zero at higher layers of the semantic encoder, leading to over-smoothing. In the context of cross-modal retrieval systems \cite{chen2019cross}, semantic inconsistency is often considered analogous to the over-smoothing issue.

\section{Methodology}
\label{Methodology}
We present DESAlign, a unified framework for robust semantic learning in MMEA that ensures semantic consistency. As discussed in Section \ref{introduction}, semantic inconsistency is prevalent in MMEA, arising from missing modal features. Existing methods, primarily addressing single modality, struggle when faced with a high rate of missing modal attributes.

To address these challenges, we delve into the semantic consistency issue from the perspective of Dirichlet energy, identifying its role in causing over-smoothing during the learning of entity representations. Motivated by this, our method aims to (a) acquire a robust semantic representation of the MMKGs and (b) perform effective semantic interpolation to achieve semantic consistency. We construct a multi-modal knowledge entity representation, encompassing graph structure, relations, images, and text, to model diverse modalities within the MMKG. We then develop robust semantic learning to counteract the over-smoothing issue, mitigating semantic inconsistency during the training process. Furthermore, we design a novel interpolation method that achieves semantic consistency by controlling the gradient flow of Dirichlet energy. Figure \ref{DESAlign} illustrates the DESAlign framework.

\begin{figure*}[t]
\centering
\includegraphics[width = \textwidth]{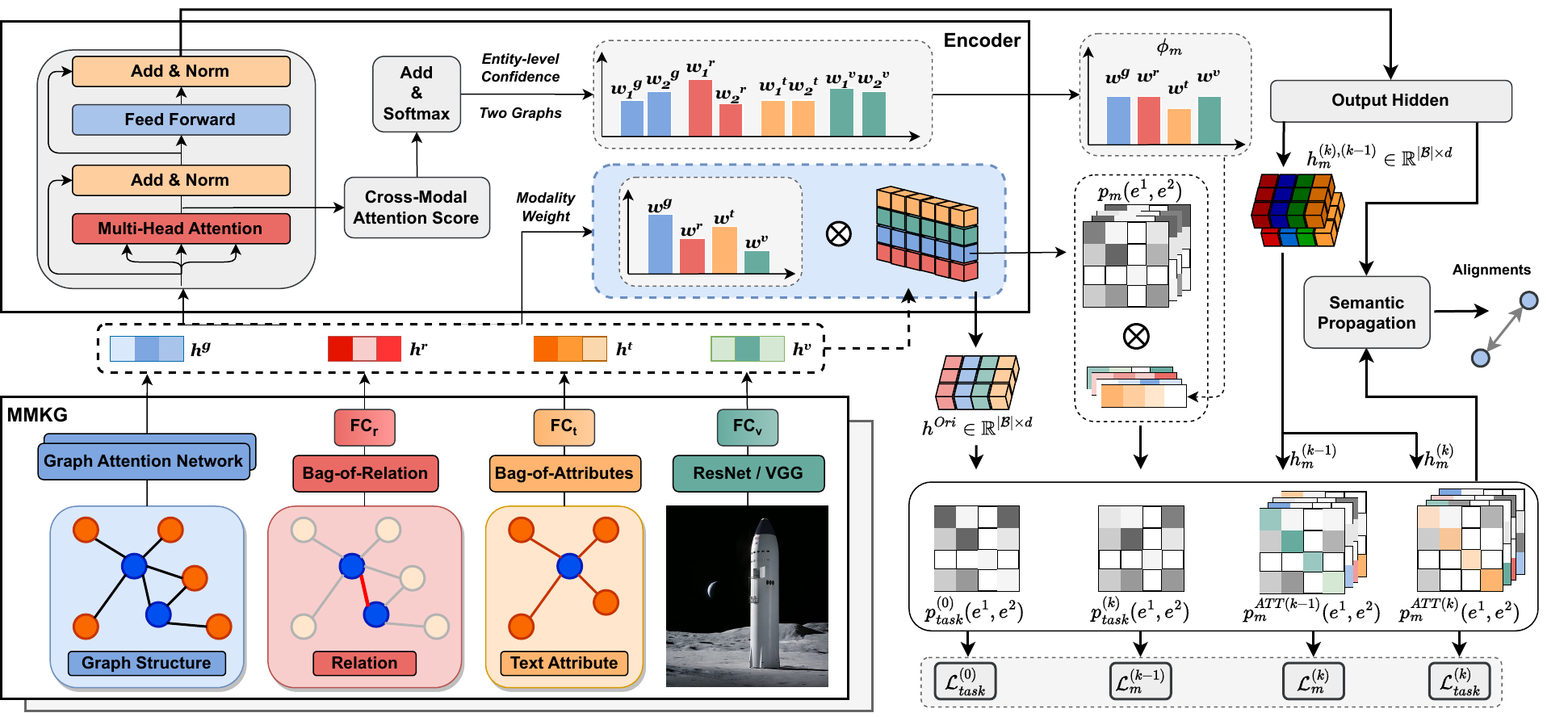}
\caption{The framework of DESAlign.} 
\vspace{-5mm}
\label{DESAlign}
\end{figure*}

\subsection{Multi-modal Knowledge Graph Representation}
This section details how we embed each modality of an entity into a low-dimensional vector within the given MMKGs.

\noindent
\textbf{(1) Graph Structure Embedding: }
Let $x^g_i \in \mathbb{R}^d$ represent the randomly initialized graph embedding of entity $e_i$, where $d$ is the predetermined hidden dimension. We utilize the Graph Attention Network (GAT)\cite{velivckovic2018graph} with two attention heads and two layers to capture the structural information of $G$, equipped with a diagonal weight matrix \cite{yang2015embedding} $W_g \in \mathbb{R}^{d\times d}$ for linear transformation. We have structure embedding as 
\begin{equation}
    h^g_i = GAT(W_g, \mathbf{A}; x^g_i)
\end{equation}

\noindent
\textbf{(2) Relation, Text Attribute, and Visual Embedding: }
To mitigate information pollution in GNNs, separate fully connected layers (FC) transform features from relations, text attributes, and vision modalities, represented as 
\begin{equation}
    h^m_i = FC_m(W_m,x_i^m), m\in \{ r,t,v \}
\end{equation}
where $r,t,v$ represent the relation, text attribute and vision modality, respectively. $x_i^m\in \mathbb{R}^{d_m}$ is the input feature of entity $e_i$ for the corresponding modality $m$. 
We follow \cite{yang2019aligning} to adopt bag-of-words for relation $(x^r)$ and text $(x^t)$, and pre-trained visual models for visual embeddings $(x^v)$. Entities lacking modal features receive randomly generated initial features, based on the distribution of existing modal features.

\noindent
\textbf{(3) Joint Semantic Embedding: }
To capture the finer-grained semantic information, we take modal-level semantic embedding to jointly represent the entity. Inspired by the MEAformer\cite{meaformer}, which follows the vanilla transformer \cite{vaswani2017attention}, we adapt Cross-modal Attention Weighted (CAW) for two types of sub-layers: the multi-head cross-modal attention block and the fully connected feed-forward networks. The attention block operates its attention function across $N_h$ parallel heads. The $i$-th head is parameterized by modally shared matrices $W_q^{(i)},W_k^{(i)},W_v^{(i)}$, transforming the multi-modal input $h_m$ into modal-aware query $h^mW_q^{(i)}$, key $h^mW_k^{(i)}$, and value $h^mW_v^{(i)}$ in $\mathbb{R}^{d_h}(d_h=d/H_h)$. It generates the following output for a given feature of modality $m$ in $\mathcal{M}=\{ g,r,t,v \}$:
\begin{equation}
    h^{ATT}_m=\bigoplus_{i=1}^{N_h}\sum_{j\in \mathcal{M}}\beta_{mj}^{(i)}h^mW_v^{(i)}\cdot W_o
\end{equation}
where $W_o \in \mathbb{R}^{d\times d}$ and $\bigoplus$ refers to the vector concatenation operation. The attention weight $(\beta_{mj})$ between an entity’s modality $m$ and $j$ in each head is formulated below:
\begin{equation}
    \beta_{mj} = \frac{exp((h^mW_q^{(i)})^\top h^jW_k^{(i)}/\sqrt{d_h})}{\sum_{n\in \mathcal{M}}exp((h^mW_q^{(i)})^\top h^nW_k^{(n)}/\sqrt{d_h})}
\end{equation}
Besides, layer normalization (LN) and residual connection (RC) are incorporated to stabilize training:
\begin{equation}
    \hat{h}^{ATT}_m = LN(h^{ATT}_m + h^m)
\end{equation}
The feed-forward network consists of two linear transformation layers and a ReLU activation function with LN and RC applied afterwards:
\begin{equation}
    \hat{h}^{ATT}_m = LN(ReLU(\hat{h}^{ATT}_mW_1+b_1)W_2+b_2+\hat{h}^{ATT}_m)
\end{equation}
where $W_1 \in \mathcal{R}^{d\times d_{in}}$ and $W_2 \in \mathcal{R}^{d_{in}\times d}$. Notably, we define the modal-level confidence $\tilde{w}^{m}$ for each modality $m$ as:
\begin{equation}
    \tilde{w}^{m}=\frac{\exp \left(\sum_{j \in \mathcal{M}} \sum_{i=0}^{N_{h}} \beta_{m j}^{(i)} / \sqrt{|\mathcal{M}| \times N_{h}}\right)}{\sum_{k \in \mathcal{M}} \exp \left(\sum_{j \in \mathcal{M}} \sum_{i=0}^{N_{h}} \beta_{k j}^{(i)} \sqrt{|\mathcal{M}| \times N_{h}}\right)}
\end{equation}
which captures crucial inter-modal interface information, and adaptively adjusts model’s cross-modal alignment preference for different modalities from each entity.

Let $\tilde{w}^{m}$ be the meta weight of entity $e_i$ for modality $m$; we formulate the joint embedding as:
\begin{equation}
    h_i^{Ori}= \bigoplus_{m\in \mathcal{M}}[\tilde{w}^{m}h_i^m], h_i^{Fus}= \bigoplus_{m\in \mathcal{M}}[\tilde{w}^{m}(\hat{h}^{ATT}_m)_i]
\end{equation}
where $h_i^{Ori}$ and $h_i^{Fus}$ are defined as the early and late fusion embedding, respectively. As a common paradigm for modality fusion in previous EA works \cite{m5,m6,mclea}, concatenation operation could be employed to prevent obscuring and over-smoothing modal diversities. We select $h_i^{Ori}$ as the final entity representation for evaluations according to our experiments, and we speculate that the modality specificity in late fusion embedding is consistently attenuated by the Transformer Layer \cite{zhou2021deepvit} which diminishes the distinction between entities.

\subsection{Multi-modal Semantic Learning}
\label{Multi-modal Semantic Learning}
Building upon Proposition \ref{pro2}, we establish the lower and upper limits of Dirichlet energy in every layer. This reveals that the quality of learning semantic representations is directly influenced by the preceding semantic information, limiting the quality from the perspective of Dirichlet energy.
\begin{myTh}
\label{MMSLearning}
    In the $k$-th layer of the semantic encoder, the lower and upper bounds are determined by $(k-1)$-th and initial semantic embedding, respectively. The training process equals optimizing the following problem:
    \begin{equation}
    \begin{aligned}
    \label{MMSLearningequation}
        \min &\quad \mathcal{L}_{task}^{(0)}+ \mathcal{L}_{task}^{(k)} +\sum_{m\in \mathcal{M}}(\mathcal{L}_m^{(k-1)} + \mathcal{L}_m^{(k)}) \\
        s.t. &\quad c_{min}\mathscr{L}(\mathbf{X}^{(k-1)}) \le \mathscr{L}(\mathbf{X}^{(k)}) \le c_{max}\mathscr{L}(\mathbf{X}^{(0)})
    \end{aligned}
    \end{equation}
    $\mathcal{L}_{task}$ denotes the cross-entropy loss of entity alignment task, $\mathcal{L}_m$ represents the intra-modal entropy from modality $m$, both $c_{min}$ and $c_{max}$ are positive hyperparameters..
\end{myTh}
\begin{proof}
From interval $(0, 1)$, there exists hyperparameter $c_{min}$ satisfying constraint of $\mathscr{L}(\mathbf{X}^{(k)})\ge c^k_{min}\mathscr{L}(\mathbf{X}^{(k-1)}) > 0$. In such a way, the over-smoothing caused by the semantic inconsistency issue is overcome since the Dirichlet energies of all the layers are larger than appropriate limits related to $c^k_{min}$. As the left part of the inequality emphasizes the semantic consistency caused by every modality, we need to pay attention to modal features in layer $k$ and $k-1$.

Compared with the initial transformed feature $\mathbf{X}^{(0)}$, the intermediate node embeddings of the same class are expected to be merged closely to have a smaller Dirichlet energy and facilitate the downstream applications. Therefore, we exploit the upper limit $c_{max}\mathscr{L}(\mathbf{X}^{(0)})$ to avoid over-separating, which emphasizes the joint initial and final semantic feature.
\end{proof}

To be specific, we introduce a unified entity alignment contrastive learning framework, inspired by \cite{mclea}, to consolidate the training objectives of the modules. For each entity pair $(e^1_i ,e^2_i)$ in $S$, we define $N_i……{ng} = \{ e^1_j|\forall e^1_j\in \mathcal{E}_1, j\neq i \} \cup \{ e^2_j|\forall e^2_j\in \mathcal{E}_2, j\neq i \}$ as its negative entity set. To improve efficiency, we adopt the in-batch negative sampling strategy \cite{chen2020simple}, restricting the sampling scope of $N_i^{ng}$ to the mini-batch $\mathcal{B}$. The alignment probability distribution is defined as:
\begin{equation}
\label{distribution}
    p_m(e_1^i ,e_2^i)=\frac{\gamma_m(e^1_i ,e^2_i)}{\gamma_m(e^1_i ,e^2_i) + \sum_{e_j\in N_i^{ng}}\gamma_m(e^1_i ,e_j)}
\end{equation}
where $\gamma_m(e_i ,e_j) =exp((h_i^{Ori})^\top, h_j^{Ori}/\tau)$ and $\tau$ is the temperature hyperparameter.  Considering the alignment direction for entity pairs reflected in equation \ref{distribution}, we define the bi-directional alignment objective for each modality $m$ as:
\begin{equation}
    \mathcal{L}_m = -\log (\phi_m(e_1^i ,e_2^i)*(p_m(e_1^i ,e_2^i) + p_m(e_2^i ,e_1^i))) /2
\end{equation}
Considering the symmetric nature of MMEA and the varying quality of aligned entities and their modal features within each KG, we employ the minimum confidence value to minimize errors. For example, $e^1_i$ may possess high-quality image and many text attributes while $e^2_i$ lacks image and text information. In such cases, using the original objective for feature alignment will inadvertently align meaningful features with random noise, thereby disrupting the encoder training process. To mitigate this issue, we define $\phi_m(e_1^i ,e_2^i)$ as the minimum confidence value for entities $e^1_i$ and $e^2_i$ in modality $m$, calculated by $\phi_m(e_1^i ,e_2^i)=Min(\tilde{w}^{m}_i,\tilde{w}^{m}_j)$.

Then we denote the joint training objective $\mathcal{L}_{task}^{(0)}$ and $\mathcal{L}_{task}^{(k)}$ when using the original and fused joint semantic embeddings, $h_i^{Ori}$ and $h_i^{Fus}$, respectively. Here, $\gamma_{\{Ori,Fus\}}(e_i ,e_j)$ is set to $exp(h_i^{\{Ori,Fus\} ^\top} h_j^{\{Ori,Fus\}}/\tau)$ and $\phi_m = 1$.

\subsection{Semantic Propagation}
\label{Semantic Propagation}
While the Dirichlet energy is convex and it is possible to derive its optimal solution in a closed-form, as discussed in \cite{rossi2022unreasonable}, its computational complexity renders it impractical for large-scale graphs, especially when numerous entities lack certain modal features. Instead, we turn our attention to the associated gradient flow $\mathbf{\bar{x}(t)} = -\nabla \mathscr{L}(\mathbf{x(t)})$. This flow can be treated as a propagation equation with the boundary condition $\mathbf{x_c(t)} = \mathbf{x_c}$. We posit that missing modal features can be expressed by existing modal information. The solution at the missing features, denoted as $\mathbf{x}_{o2} = \lim_{t\to \infty} \mathbf{(\mathbf{x}_c(t),\mathbf{x}_{o1}(t)})$, provides the desired interpolation.

The gradient of Dirichlet energy is $\nabla_{\mathbf{x}}\mathscr{L} = \mathbf{\Delta}\mathbf{x}$, and the gradient flow has the propagation form with initial and boundary conditions:
\begin{equation*}
    \mathbf{\bar{x}(t)} = -\mathbf{\Delta}\mathbf{x},
    \mathbf{x(0)} = 
    \left [ 
    \begin{array}{l}
        \mathbf{x}_c \\
        \mathbf{x}_{o1}\\
        \mathbf{x}_{o2}\\
    \end{array}
\right ],
    \mathbf{x_c(t)} = \mathbf{x_c}
\end{equation*}
By incorporating the boundary conditions, the propagation equation can be compactly expressed as:
\begin{equation}
\begin{aligned}
\label{propagation equation}
\left [ 
    \begin{array}{c}
        \bar{\mathbf{x}}_c(t) \\
        \bar{\mathbf{x}}_{o1}(t)\\
        \bar{\mathbf{x}}_{o2}(t)\\
    \end{array}
\right ] & = 
-\left [ 
    \begin{array}{ccc}
        \mathbf{0} & \mathbf{0} & \mathbf{0}\\
        \mathbf{\Delta}_{o1c} & \mathbf{\Delta}_{o1o1} & \mathbf{\Delta}_{o1o2}\\
        \mathbf{\Delta}_{o2c} & \mathbf{\Delta}_{o2o1} & \mathbf{\Delta}_{o2o2}\\
    \end{array}
\right ]
\left [ 
    \begin{array}{c}
        {\mathbf{x}}_c \\
        {\mathbf{x}}_{o1}(t)\\
        {\mathbf{x}}_{o2}(t)\\
    \end{array}
\right ] \\
& = -\left [ 
    \begin{array}{c}
        \mathbf{0}\\
        \mathbf{\Delta}_{o1c}{\mathbf{x}}_c + \mathbf{\Delta}_{o1o1}{\mathbf{x}}_{o1}(t) + \mathbf{\Delta}_{o1o2}{\mathbf{x}}_{o2}(t)\\
        \mathbf{\Delta}_{o2c}{\mathbf{x}}_c + \mathbf{\Delta}_{o2o1}{\mathbf{x}}_{o1}(t) + \mathbf{\Delta}_{o2o2}{\mathbf{x}}_{o2}(t)\\
    \end{array}
\right ]    
\end{aligned}
\end{equation}
As expected, the gradient flow of the semantic-consistent features is $\mathbf{0}$, given that they do not change and exhibit stationary properties during propagation. The evolution of the missing features $\mathbf{x}_{o2}$ can be regarded as a propagation equation with a constant source $\mathbf{\Delta}_{o2c}{\mathbf{x}}_c$ coming from the boundary (semantic-consistent) nodes. Then we have the following proposition.
\begin{myTh}
\label{pro4}
    The absent modal features from the missing modality can be expressed by existing modal features in the process of propagation: ${\mathbf{x}}_{o2} = f(\mathbf{x}_{c}, \mathbf{x}_{o1})$
\end{myTh}
\begin{proof}
Since the Dirichlet energy $\mathscr{L}$ is convex, its global minimizer is given by the solution to the closed-form equation $\nabla_{\mathbf{x}_{o2}}\mathscr{L} = 0$ and by rearranging the final rows of (\ref{propagation equation}), we get the solution:
\begin{equation}
    {\mathbf{x}}_{o2}(t) = -\mathbf{\Delta}_{o2o2}^{-1}\mathbf{\Delta}_{o2c}{\mathbf{x}}_c -\mathbf{\Delta}_{o2o2}^{-1}\mathbf{\Delta}_{o2o1}{\mathbf{x}}_{o1}(t)
\end{equation}
The sub-Laplacian matrix of an undirected connected graph is invertible \cite{rossi2022unreasonable}, thus $\mathbf{\Delta}_{o2o2}$ is non-singular, ensuring the existence of this solution.
\end{proof}
However, solving a system of linear equations is computationally expensive (incurring $O(|\mathcal{E}_0|^3)$ complexity for matrix inversion) and thus intractable for anything but only small graphs. As an alternative, we can discretize the propagation equation (\ref{propagation equation}) and solve it by an iterative numerical scheme. Approximating the temporal derivative with a forward difference with the time variable $t$ discretized using a fixed step ($t=hk$ for step size $h > 0$ and $k = 1,2,\dots$), we obtain the explicit Euler scheme:
\begin{equation}
\begin{aligned}
    \mathbf{x}^{(k+1)}  = \mathbf{x}^{(k)}  - h
    \left [ 
    \begin{array}{ccc}
        \mathbf{0} & \mathbf{0} & \mathbf{0}\\
        \mathbf{\Delta}_{o1c} & \mathbf{\Delta}_{o1o1} & \mathbf{\Delta}_{o1o2}\\
        \mathbf{\Delta}_{o2c} & \mathbf{\Delta}_{o2o1} & \mathbf{\Delta}_{o2o2}\\
    \end{array}
\right ]\mathbf{x}^{(k)} \\
 = \left [ 
    \begin{array}{ccc}
        \mathbf{I} & \mathbf{0} & \mathbf{0}\\
        -h\mathbf{\Delta}_{o1c} & \mathbf{I} -h\mathbf{\Delta}_{o1o1} & -h\mathbf{\Delta}_{o1o2}\\
        -h\mathbf{\Delta}_{o2c} & -h\mathbf{\Delta}_{o2o1} & \mathbf{I} - h\mathbf{\Delta}_{o2o2}\\
    \end{array}
\right ]\mathbf{x}^{(k)}
\end{aligned}
\end{equation}
For the special case of $h = 1$, we can use the definition of Laplacian matrix $\mathbf{\Delta} = \mathbf{I} - \mathbf{\widetilde{A}}$ to rewrite the iteration formula:
\begin{equation}
\label{filter}
    \mathbf{x}^{(k+1)} = \left [ 
    \begin{array}{ccc}
        \mathbf{I} & \mathbf{0} & \mathbf{0}\\
        \mathbf{\widetilde{A}}_{o1c} & \mathbf{\widetilde{A}}_{o1o1} & \mathbf{\widetilde{A}}_{o1o2}\\
        \mathbf{\widetilde{A}}_{o2c} & \mathbf{\widetilde{A}}_{o2o1} & \mathbf{\widetilde{A}}_{o2o2}\\
    \end{array}
\right ]\mathbf{x}^{(k)}
\end{equation}
The Euler scheme is the gradient descent of the Dirichlet energy. Thus, applying the scheme decreases the Dirichlet energy, resulting in the features becoming increasingly smooth. Simultaneously, it utilizes existing semantic information to fill in the missing part of the semantic features. Iteration \ref{filter} can be interpreted as successive low-pass filtering. Our experiments indicate that the higher the rate of missing features, the stronger the low-pass filtering effect.

We note that the update in Equation \ref{filter} is equivalent to first multiplying the feature vector $\mathbf{x}$ by the matrix $\mathbf{\widetilde{A}}$, and then resetting the known features to their true value. This yields an extremely simple and scalable iterative algorithm, outlined in Algorithm \ref{Algorithm}, to reconstruct the missing modal features in the semantic space, which we refer to as \textit{Semantic Propagation}.
\begin{equation}
    \mathbf{x}^{(k+1)} \leftarrow \mathbf{\widetilde{A}}\mathbf{x}^{(k)}; \mathbf{x}^{(k+1)}_c \leftarrow \mathbf{x}^{(k)}_c
\end{equation}
We observe that features $\mathbf{x}_{o2}$ take on distinct calculation forms across iteration rounds, indicating varying semantic information in $\mathbf{x}_{o2}$. 
To maximize the utilization of this information and preserve the original distribution of consistent features $\mathbf{x}_c$, we calculate the mean of pairwise similarities from all propagation processes as the final pairwise similarities.

\begin{algorithm}[t]
\caption{DESAlign}
\label{Algorithm}
\KwIn{\\
$\Phi^\prime$: the seed alignments of two MMKGs. \\
$\mathbf{A_s}, \mathbf{A_t}$: the adjacency matrices of MMKGs.\\
$x^g$: the randomly initialized structural features.\\
$x^r, x^t, x^v$: relation, text attribute and image features.\\}
\KwOut{The pairwise similarity $\Omega$.}
Initialize the modality set $\mathcal{M}=\{ g,r,t,v \}$\;
Normalize the adjacency matrices $\mathbf{A_s}$, $\mathbf{A_t}$ as $\mathbf{\widetilde{A}_s}, \mathbf{\widetilde{A}_t}$\;
\For{$k\ \mathrm{in}\ \mathrm{range}(max\_epoch)$}{
$h^g \leftarrow$ calculate the structure embedding for each entity $e_i$ by using $h^g_i = GAT(\mathbf{A_s}/\mathbf{A_t}; x^g_i)$\;
$h^m, m\in \{ r,t,v \} \leftarrow$ calculate the modal embedding for each entity $e_i$ by using $h^m_i = FC_m(x_i^m)$\;
$\tilde{w}^{m},\hat{h}^m, m\in \mathcal{M} \leftarrow$ calculate the cross-modal weight and embedding by using $\tilde{w}^{m},\hat{h}^m = CAW(h^m)$\;
$\mathbf{X}^{(0)} \leftarrow$ generate the initial semantic embedding by concatenating weighted features $[\tilde{w}^{m}h_i^m]$\;
$\mathbf{X}^{(k)}, \mathbf{X}^{(k-1)} \leftarrow$ generate the semantic embedding by concatenating weighted features $[\tilde{w}^{m}\hat{h}_i^m]$\;
$Loss \leftarrow$ define the loss according to Equation \ref{MMSLearningequation}\;
Update parameters of DESAlign via back-propagating $Loss$\;
\For{$j$ in range($max\_iter$)}{
$\mathbf{X}_{s/ t} \leftarrow \mathbf{\widetilde{A}}_{s/ t}\mathbf{X}_{s/ t}$\;
$\Omega_j \leftarrow$ calculate the pairwise similarity between $\mathbf{X_s}$ and $\mathbf{X_t}$, $\mathbf{X} = \mathbf{X_s} \cup \mathbf{X_t}$\;
}
Generate $\Omega$ by $\Omega = \sum_{j=0}^{max\_iter}\Omega_j / {max\_iter}$
}
\end{algorithm}

\section{Experiments}
\label{MAINExperiments}
In this section, we assess the effectiveness of the proposed DESAlign. We aim to answer the following questions:

\begin{itemize}
\item (Q1) Is DESAlign more robust to semantic inconsistency compared with other alignment methods?
\item (Q2) Does DESAlign outperform the state-of-the-art methods on real-world graph alignment applications?
\item (Q3) Does the multi-modal semantic learning and semantic propagation in DESAlign genuinely contribute to the alignment performance?
\item (Q4) What is the performance of DESAlign and other prominent methods in weakly supervised settings?
\item (Q5) How does semantic propagation perform with different iterations?
\end{itemize}

Our code and data are provided in the supplementary materials and will be publicly available upon acceptance.

\begin{table}[t!]
\caption{Statistics for datasets. EA pairs refer to the seed aligments}
\label{datasets}
\Huge
\renewcommand\arraystretch{1.5}
\setlength{\tabcolsep}{3pt}
\centering
\resizebox{1\linewidth}{!}{
\begin{tabular}{c|c|ccccccc}
\toprule[1.5pt]
Dataset                        & KG            & Ent.   & Rel.  & Att.  & R.Triples & A.Triples & Image  & EA pairs                \\ 
\midrule[1.5pt]
\multirow{2}{*}{FB-DB}   & FB15K         & 14,951 & 1,345 & 116   & 592,213     & 29,395      & 13,444 & \multirow{2}{*}{12,846} \\
                               & DB15K         & 12,842 & 279   & 225   & 89,197      & 48,080      & 12,837 &                         \\ \midrule
\multirow{2}{*}{FB-YAGO} & FB15K         & 14,951 & 1,345 & 116   & 592,213     & 29,395      & 13,444 & \multirow{2}{*}{11,199} \\
                               & YAGO15K       & 15,404 & 32    & 7     & 122,886     & 23,532      & 11,194 &                         \\ \midrule
\multirow{2}{*}{\begin{tabular}[c]{@{}c@{}}DBP15K\\ ZH-EN\end{tabular}}   & ZH(Chinese)   & 19,388 & 1,701 & 8,111 & 70,414      & 248,035     & 15,912 & \multirow{2}{*}{15,000} \\
                               & EN (English)  & 19,572 & 1,323 & 7,173 & 95,142      & 343,218     & 14,125 &                         \\ \midrule
\multirow{2}{*}{\begin{tabular}[c]{@{}c@{}}DBP15K\\ JA-EN\end{tabular}}   & JA (Japanese) & 19,814 & 1,299 & 5,882 & 77,214      & 248,991     & 12,739 & \multirow{2}{*}{15,000} \\
                               & EN (English)  & 19,780 & 1,153 & 6,066 & 93,484      & 320,616     & 13,741 &                         \\ \midrule
\multirow{2}{*}{\begin{tabular}[c]{@{}c@{}}DBP15K\\ FR-EN\end{tabular}}    & FR (French)   & 19,661 & 903   & 4,547 & 105,998     & 273,825     & 14,174 & \multirow{2}{*}{15,000} \\
                               & EN (English)  & 19,993 & 1,208 & 6,422 & 115,722     & 351,094     & 13,858 &                         \\ 
\bottomrule[1.5pt]
\end{tabular}}
\end{table}

\subsection{Experimental Settings}
We first introduce the experimental settings, covering datasets, baseline methods, evaluation metrics, and implementation details.

\subsubsection{\textbf{Datasets}}
We list the statistics of all datasets used in the experiments in Table \ref{datasets}. We consider two types of datasets:
\begin{itemize}
    \item Monolingual: We select FB15K-DB15K (FBDB15K) and FB15K-YAGO15K (FBYG15K) from MMKG \cite{liu2019mmkg} with three data splits: $R_{seed} \in \{0.2, 0.5, 0.8\}$.
    \item Bilingual: DBP15K \cite{sun2017cross} contains three datasets built from the multilingual versions of DBpedia, including DBP15K$_{FR-EN, JA-EN, ZH-EN}$. Each contains about 400K triples and 15K seed alignments with 30\% of them as seed alignments ($R_{seed} = 0.3$). We adopt their multi-model variant \cite{m6} with images attached to the entities. 
\end{itemize}
Additionally, we collect variants of these datasets with different degrees of semantic inconsistency to analyze and validate the robustness of DESAlign and baselines. Specifically, we set $R_{seed}$ of FBDB15K and DBP15K$_{FR-EN}$ from 1\% to 30\% to analyze the method’s robustness and effectiveness with different supervised settings. We also set $R_{img}$ (image ratio) of DBP15K and $R_{tex}$ (attribute ratio) of FBDB15K and FBYG15K from 5\% to 60\% to validate robustness against semantic inconsistency.

\subsubsection{\textbf{Baseline}}
We compare DESAlign with 18 entity alignment baselines, including twelve methods based on the basic model (TransE\cite{bordes2013translating}, IPTransE\cite{zhu2017iterative}, GCN-align\cite{wang2018cross}, SEA\cite{pei2019semi}, IMUSE\cite{he2019unsupervised}, AttrGNN\cite{liu2020exploring}, PoE, PoE-rni\cite{liu2019mmkg}, MMEA\cite{chen2020mmea}, HEA\cite{m2}, ACK\cite{ack}, MUGCN\cite{cao2019multi}, ALiNET\cite{sun2020knowledge}) and six prominent methods with an iterative strategy (BootEA\cite{sun2018bootstrapping}, NAEA\cite{zhu2019neighborhood}, EVA\cite{m6}, MSNEA\cite{chen2022multi}, MCLEA\cite{mclea}, MEAformer\cite{meaformer}). The iterative strategy, following \cite{mclea}, can be viewed as a buffering mechanism, which maintains a temporary cache to store cross-graph mutual nearest entity pairs from the testing set. While UMAEA \cite{UMAEA} delves into addressing visual ambiguity and missing data issues, its adoption of a multi-stage pipeline renders it non-end-to-end, and hence, we exclude it from consideration in our comparative study.

\subsubsection{\textbf{Evaluation Metrics} }
We compute entity similarities using cosine similarity and assess performance with H@k and MRR metrics. H@k measures the proportion of correctly aligned entities ranked within the top $k$ positions:
\begin{equation}
H@k = \frac{1}{|S_t|}\sum_{i=1}^{|S_t|}\mathbb{I}[rank_{i}\le k]
\end{equation}
where $rank_i$ is the ranking of the first accurate alignment for query entity $i$, $\mathbb{I}$ indicates correctness ($\mathbb{I}=1$ if $rank_{i}\le k$, else $0$), and $S_t$ is the set of test alignments. MRR evaluates the average of reciprocal ranks of the first correct answer for queries, reflecting the model's precision across a set of queries:
\begin{equation}
MRR = \frac{1}{|S_t|}\sum_{i=1}^{|S_t|}\frac{1}{rank_i}
\end{equation}

\subsubsection{\textbf{Implementation Details}}
For all mentioned baselines, we run the code provided by the authors and keep the default configuration. In DESAlign, we apply the following settings:
(\romannumeral1) The hidden layer dimensions $d$ for all networks are unified into 300. The total epochs for baselines are set to 500 with an optional iterative training strategy applied for another 500 epochs, following \cite{mclea}. Training strategies including cosine warm-up schedule (15\% steps for LR warm-up), early stopping, and gradient accumulation are adopted.
(\romannumeral2) The AdamW optimizer ($\beta_1 = 0.9$, $\beta_2 = 0.999$) is used, with a fixed batch size of 3500.
(\romannumeral3) To demonstrate model stability, following \cite{chen2020mmea, mclea}, the vision encoders are set to ResNet-152 \cite{he2016deep} where the vision feature dimension $d_v$ is 2048.
(\romannumeral4) An alignment editing method is employed to reduce error accumulation \cite{sun2018bootstrapping}.
(\romannumeral5) Following \cite{yang2019aligning}, Bag-of-Words (BoW) is selected for encoding relations ($x^r$) and attributes ($x^a$) as fixed-length (i.e., $d_r$ = $d_a = 1000$) vectors.
(\romannumeral6) $\tau$ is set to 0.1, determining how much attention the contrast loss pays to difficult negative samples, and the head number $N_h$ in multi-head attention is set to 1.
Despite potential performance variations resulting from parameter searching, our focus remained on achieving broad applicability rather than fine-tuning for specific datasets. During iterative training, the pipeline is repeated. All experiments are performed on a computing server running Ubuntu 22.04 with an Intel Xeon Silver 4210R CPU and two NVIDIA A100 GPUs.

\subsection{(Q1) Alignment over Semantic Inconsistency}

\begin{table*}[t!]
\caption{Main results of Prominent Methods with Varying Ratio of Text Attributes.}
\label{Varying Text}
\renewcommand\arraystretch{1}
\setlength{\tabcolsep}{5pt}
\centering
\large
\vspace{-0.2cm}
\resizebox{\textwidth}{!}{
\begin{tabular}{cc|cccccccccccccccccc}
\toprule[1.2pt]
\multicolumn{2}{c|}{Text Ratio}                                                                                      & \multicolumn{3}{c}{$R_{tex} = 5\%$}                  & \multicolumn{3}{c}{$R_{tex} = 20\%$}                   & \multicolumn{3}{c}{$R_{tex} = 30\%$}                   & \multicolumn{3}{c}{$R_{tex} = 40\%$}                   & \multicolumn{3}{c}{$R_{tex} = 50\%$}                   & \multicolumn{3}{c}{$R_{tex} = 60\%$}    \cr
 \cmidrule(lr){1-2}\cmidrule(lr){3-5}\cmidrule(lr){6-8}\cmidrule(lr){9-11}\cmidrule(lr){12-14} \cmidrule(lr){15-17} \cmidrule(lr){18-20} 
\multicolumn{1}{c|}{Datasets}                             & Model            & H@1           & H@10          & MRR           & H@1           & H@10          & MRR           & H@1           & H@10          & MRR           & H@1           & H@10          & MRR           & H@1           & H@10          & MRR           & H@1           & H@10          & MRR           \\ 
\midrule
\multicolumn{1}{c|}{\multirow{5}{*}{\rotatebox{90}{\begin{tabular}[c]{@{}c@{}}FB15K\\ DB15K\end{tabular}}}}    & EVA              & 9.4           & 29.2          & 16.1          & 9.3           & 29.2          & 15.9          & 9.3           & 29.2          & 16.0          & 9.1           & 28.9          & 15.8          & 9.2           & 29.0          & 15.9          & 9.3           & 29.0          & 16.0          \\
\multicolumn{1}{c|}{}                                     & MCLEA            & 33.8          & 63.8          & 44.0          & 33.5          & 63.4          & 43.8          & 33.4          & 63.3          & 43.6          & 32.7          & 63.0          & 43.0          & 32.3          & 62.1          & 42.5          & 32.0          & 61.7          & 42.2          \\
\multicolumn{1}{c|}{}                                     & MEAformer        & 36.4          & 67.8          & 47.1          & 33.8          & 66.1          & 44.7          & 33.6          & 64.9          & 44.2          & 33.8          & 64.8          & 44.2          & 33.4          & 64.9          & 44.0          & 33.6          & 65.4          & 44.3          \\
\cmidrule{2-20}
\multicolumn{1}{c|}{}                                     & \textbf{Ours}    & \textbf{47.1} & \textbf{73.5} & \textbf{56.3} & \textbf{46.8} & \textbf{73.4} & \textbf{56.1} & \textbf{46.6} & \textbf{73.5} & \textbf{56.0} & \textbf{46.8} & \textbf{73.6} & \textbf{56.1} & \textbf{47.2} & \textbf{73.8} & \textbf{56.3} & \textbf{47.2} & \textbf{74.0} & \textbf{56.5} \\
\multicolumn{1}{c|}{}                                     & \textbf{Impprov.} & \textbf{10.7} & \textbf{5.7}  & \textbf{9.2}  & \textbf{13.0} & \textbf{7.3}  & \textbf{11.4} & \textbf{13.0} & \textbf{8.6}  & \textbf{11.8} & \textbf{13.0} & \textbf{8.8}  & \textbf{11.9} & \textbf{13.8} & \textbf{8.9}  & \textbf{12.3} & \textbf{13.6} & \textbf{8.6}  & \textbf{12.2} \\
\midrule
\multicolumn{1}{c|}{\multirow{5}{*}{\rotatebox{90}{\begin{tabular}[c]{@{}c@{}}FB15K\\ YAGO15K\end{tabular}}}} & EVA              & 8.8           & 24.9          & 14.3          & 8.8           & 25.1          & 14.4          & 8.7           & 24.9          & 14.4          & 8.8           & 24.9          & 14.4          & 8.9           & 25.2          & 14.5          & 8.9           & 25.2          & 14.5          \\
\multicolumn{1}{c|}{}                                     & MCLEA            & 29.4          & 54.6          & 38.0          & 29.5          & 54.7          & 38.0          & 29.3          & 54.3          & 37.9          & 29.4          & 54.1          & 37.9          & 29.1          & 54.2          & 37.7          & 29.0          & 54.0          & 37.6          \\
\multicolumn{1}{c|}{}                                     & MEAformer        & 30.1          & 55.1          & 39.3          & 30.2          & 55.0          & 39.5          & 30.5          & 55.0          & 39.5          & 29.6          & 54.6          & 38.9          & 29.7          & 54.3          & 38.9          & 29.2          & 54.2          & 38.5          \\
\cmidrule{2-20}
\multicolumn{1}{c|}{}                                     & \textbf{Ours}    & \textbf{39.1} & \textbf{64.2} & \textbf{47.6} & \textbf{39.0} & \textbf{65.4} & \textbf{48.1} & \textbf{39.3} & \textbf{64.7} & \textbf{48.0} & \textbf{39.4} & \textbf{64.8} & \textbf{48.0} & \textbf{39.6} & \textbf{64.8} & \textbf{48.1} & \textbf{39.5} & \textbf{64.8} & \textbf{48.2} \\
\multicolumn{1}{c|}{}                                     & \textbf{Impprov.} & \textbf{9.0}  & \textbf{9.1}  & \textbf{8.3}  & \textbf{8.8}  & \textbf{10.4} & \textbf{8.6}  & \textbf{8.8}  & \textbf{9.7}  & \textbf{8.5}  & \textbf{9.8}  & \textbf{10.2} & \textbf{9.1}  & \textbf{9.9}  & \textbf{10.5} & \textbf{9.2}  & \textbf{10.3} & \textbf{10.6} & \textbf{9.7}  \\
\bottomrule
\end{tabular}}
{\footnotesize{Remark: \textbf{bold} means best. "Improv." represents the percentage increase compared with SOTA.}}
\end{table*}

\begin{table*}[t!]
\caption{Main results of Prominent Methods with Varying Ratio of Image.}
\label{Varying Image}
\vspace{-0.2cm}
\renewcommand\arraystretch{1}
\setlength{\tabcolsep}{5pt}
\centering
\large
\resizebox{\textwidth}{!}{
\begin{tabular}{cc|cccccccccccccccccc}
\toprule[1.2pt]
\multicolumn{2}{c|}{Image Ratio}                                                                                      & \multicolumn{3}{c}{$R_{img} = 5\%$}                  & \multicolumn{3}{c}{$R_{img} = 20\%$}                   & \multicolumn{3}{c}{$R_{img} = 30\%$}                   & \multicolumn{3}{c}{$R_{img} = 40\%$}                   & \multicolumn{3}{c}{$R_{img} = 50\%$}                   & \multicolumn{3}{c}{$R_{img} = 60\%$}    \cr
 \cmidrule(lr){1-2}\cmidrule(lr){3-5}\cmidrule(lr){6-8}\cmidrule(lr){9-11}\cmidrule(lr){12-14} \cmidrule(lr){15-17} \cmidrule(lr){18-20} 
\multicolumn{1}{c|}{Datasets}                                                                & Model            & H@1           & H@10          & MRR           & H@1           & H@10          & MRR           & H@1           & H@10          & MRR           & H@1           & H@10          & MRR           & H@1           & H@10          & MRR           & H@1           & H@10          & MRR           \\ 
\midrule
\multicolumn{1}{c|}{\multirow{5}{*}{\rotatebox{90}{\begin{tabular}[c]{@{}c@{}}DBP15K\\ ZH-EN\end{tabular}}}} & EVA              & 62.3          & 87.8          & 71.5          & 62.4          & 87.8          & 71.6          & 61.7          & 89.1          & 71.5          & 62.3          & 87.5          & 71.1          & 61.4          & 88.7          & 71.2          & 62.5          & 87.6          & 71.7          \\
\multicolumn{1}{c|}{}                                                                        & MCLEA            & 63.8          & 90.5          & 73.2          & 58.8          & 86.5          & 68.6          & 58.5          & 90.1          & 69.6          & 61.1          & 87.4          & 70.4          & 63.1          & 91.7          & 73.1          & 66.1          & 89.6          & 74.4          \\
\multicolumn{1}{c|}{}                                                                        & MEAformer        & 65.8          & 91.5          & 75.1          & 65.7          & 91.6          & 75.0          & 66.8          & 91.8          & 75.7          & 66.8          & 92.2          & 75.9          & 69.0          & 92.7          & 77.6          & 71.3          & 93.2          & 79.2          \\
\cmidrule{2-20}
\multicolumn{1}{c|}{}                                                                        & \textbf{DESAlign}    & \textbf{73.3} & \textbf{93.7} & \textbf{80.8} & \textbf{73.7} & \textbf{94.1} & \textbf{81.2} & \textbf{73.8} & \textbf{94.0} & \textbf{81.2} & \textbf{74.5} & \textbf{94.4} & \textbf{81.9} & \textbf{75.8} & \textbf{94.6} & \textbf{82.8} & \textbf{77.0} & \textbf{95.0} & \textbf{83.6} \\
\multicolumn{1}{c|}{}                                                                        & \textbf{Improv.} & \textbf{7.5}  & \textbf{2.2}  & \textbf{5.7}  & \textbf{8.0}  & \textbf{2.5}  & \textbf{6.2}  & \textbf{7.0}  & \textbf{2.2}  & \textbf{5.5}  & \textbf{7.7}  & \textbf{2.2}  & \textbf{6.0}  & \textbf{6.8}  & \textbf{1.9}  & \textbf{5.2}  & \textbf{5.7}  & \textbf{1.8}  & \textbf{4.4}  \\ 
\midrule
\multicolumn{1}{c|}{\multirow{5}{*}{\rotatebox{90}{\begin{tabular}[c]{@{}c@{}}DBP15K\\ JA-EN\end{tabular}}}} & EVA              & 61.5          & 87.7          & 70.8          & 61.6          & 87.7          & 71.0          & 62.0          & 88.2          & 71.3          & 61.6          & 87.8          & 71.1          & 59.5          & 88.0          & 69.5          & 62.4          & 88.1          & 71.6          \\
\multicolumn{1}{c|}{}                                                                        & MCLEA            & 59.9          & 89.7          & 70.6          & 57.9          & 84.6          & 67.5          & 56.2          & 90.5          & 68.3          & 61.3          & 86.7          & 70.3          & 61.1          & 91.7          & 72.0          & 68.6          & 89.8          & 76.1          \\
\multicolumn{1}{c|}{}                                                                        & MEAformer        & 64.1          & 92.2          & 74.3          & 64.6          & 92.3          & 74.6          & 65.4          & 92.8          & 75.3          & 67.0          & 93.0          & 76.3          & 69.3          & 93.9          & 78.2          & 72.7          & 95.0          & 80.7          \\
\cmidrule{2-20}
\multicolumn{1}{c|}{}                                                                        & \textbf{DESAlign}    & \textbf{72.9} & \textbf{94.3} & \textbf{80.9} & \textbf{73.5} & \textbf{94.7} & \textbf{81.3} & \textbf{74.0} & \textbf{94.7} & \textbf{81.7} & \textbf{75.1} & \textbf{95.1} & \textbf{82.5} & \textbf{76.6} & \textbf{95.4} & \textbf{83.5} & \textbf{78.4} & \textbf{95.8} & \textbf{84.8} \\
\multicolumn{1}{c|}{}                                                                        & \textbf{Improv.} & \textbf{8.8}  & \textbf{2.1}  & \textbf{6.6}  & \textbf{8.9}  & \textbf{2.4}  & \textbf{6.7}  & \textbf{8.6}  & \textbf{1.9}  & \textbf{6.4}  & \textbf{8.1}  & \textbf{2.1}  & \textbf{6.2}  & \textbf{7.3}  & \textbf{1.5}  & \textbf{5.3}  & \textbf{5.7}  & \textbf{0.8}  & \textbf{4.1}  \\ 
\midrule
\multicolumn{1}{c|}{\multirow{5}{*}{\rotatebox{90}{\begin{tabular}[c]{@{}c@{}}DBP15K\\ FR-EN\end{tabular}}}} & EVA              & 62.4          & 89.5          & 72.0          & 62.4          & 89.5          & 72.0          & 64.2          & 91.0          & 73.7          & 62.6          & 89.8          & 72.1          & 62.4          & 89.2          & 72.0          & 63.4          & 90.0          & 72.8          \\
\multicolumn{1}{c|}{}                                                                        & MCLEA            & 63.4          & 93.0          & 74.1          & 58.2          & 86.3          & 68.2          & 57.5          & 92.0          & 69.8          & 60.1          & 87.9          & 70.2          & 58.6          & 92.1          & 70.0          & 67.5          & 90.1          & 75.7          \\
\multicolumn{1}{c|}{}                                                                        & MEAformer        & 64.9          & 93.6          & 75.4          & 64.5          & 93.6          & 75.0          & 65.2          & 93.9          & 75.7          & 66.5          & 94.4          & 76.6          & 69.2          & 94.6          & 78.5          & 72.4          & 95.3          & 80.8          \\
\cmidrule{2-20}
\multicolumn{1}{c|}{}                                                                        & \textbf{DESAlign}    & \textbf{75.1} & \textbf{95.5} & \textbf{82.7} & \textbf{75.5} & \textbf{95.8} & \textbf{83.0} & \textbf{76.0} & \textbf{95.8} & \textbf{83.4} & \textbf{77.0} & \textbf{95.9} & \textbf{84.1} & \textbf{78.1} & \textbf{96.2} & \textbf{84.8} & \textbf{79.9} & \textbf{96.3} & \textbf{86.1} \\
\multicolumn{1}{c|}{}                                                                        & \textbf{Improv.} & \textbf{10.2} & \textbf{1.9}  & \textbf{7.3}  & \textbf{11.0} & \textbf{2.2}  & \textbf{8.0}  & \textbf{10.8} & \textbf{1.9}  & \textbf{7.7}  & \textbf{10.5} & \textbf{1.5}  & \textbf{7.5}  & \textbf{8.9}  & \textbf{1.6}  & \textbf{6.3}  & \textbf{7.5}  & \textbf{1.0}  & \textbf{5.3}  \\ 
\bottomrule[1.2pt]
\end{tabular}}
\end{table*}

\begin{table*}[t!]
\caption{Main results of monolingual datasets. \underline{underline} means runner-up.}
\label{monolingual result}
\large
\vspace{-0.2cm}
\renewcommand\arraystretch{1}
\setlength{\tabcolsep}{4pt}
\centering
\resizebox{1\textwidth}{!}{
\begin{tabular}{ccccccccccc|ccccccccc}
\toprule
\multicolumn{2}{c}{\multirow{2}{*}{Datasets}} & \multicolumn{9}{c|}{FB15K-DB15K} & \multicolumn{9}{c}{FB15K-YAGO15K} \cr
\cmidrule(lr){3-20}
\multicolumn{2}{c}{} & \multicolumn{3}{c}{$R_{seed}=$20\%}    & \multicolumn{3}{c}{$R_{seed}=$50\%}    & \multicolumn{3}{c|}{$R_{seed}=$80\%}     & \multicolumn{3}{c}{$R_{seed}=$20\%} & \multicolumn{3}{c}{$R_{seed}=$50\%} & \multicolumn{3}{c}{$R_{seed}=$80\%}  \cr
 \cmidrule(lr){1-2}\cmidrule(lr){3-5}\cmidrule(lr){6-8}\cmidrule(lr){9-11}\cmidrule(lr){12-14} \cmidrule(lr){15-17} \cmidrule(lr){18-20} 
\multicolumn{2}{c}{Models} &H@1  &H@10 &MRR &H@1  &H@10 &MRR &H@1 &H@10 &MRR &H@1  &H@10 &MRR &H@1  &H@10 &MRR &H@1  &H@10 &MRR \cr
\midrule
\multirow{15}{*}{\rotatebox{90}{Basic}}    &\multicolumn{1}{|c|}{TransE}           & 7.8           & 24.0          & 13.4          & 23.0          & 44.6          & 30.6          & 42.6          & 65.9          & 50.7          & 6.4           & 20.3          & 11.2          & 19.7          & 38.5          & 26.2          & 39.2          & 59.5          & 46.3          \cr
\multirow{15}{*}{}                           & \multicolumn{1}{|c|}{IPTransE}         & 6.5           & 21.5          & 9.4           & 21.0          & 42.1          & 28.3          & 40.3          & 62.7          & 46.9          & 4.7           & 16.9          & 8.4           & 20.1          & 36.9          & 24.8          & 40.1          & 60.2          & 45.8          \cr
\multirow{15}{*}{}                            & \multicolumn{1}{|c|}{GCN-align}        & 5.3           & 17.4          & 8.7           & 22.6          & 43.5          & 29.3          & 41.4          & 63.5          & 47.2          & 8.1           & 23.5          & 15.3          & 23.5          & 42.4          & 29.4          & 40.6          & 64.3          & 47.7          \cr
\multirow{15}{*}{}                            & \multicolumn{1}{|c|}{SEA}              & 17.0          & 42.5          & 25.5          & 37.3          & 65.7          & 47.0          & 51.2          & 78.4          & 50.5          & 14.1          & 37.1          & 21.8          & 29.4          & 57.7          & 38.8          & 51.4          & 77.3          & 60.5          \cr
\multirow{15}{*}{}                           & \multicolumn{1}{|c|}{IMUSE}            & 17.6          & 43.5          & 26.4          & 30.9          & 57.6          & 40.0          & 45.7          & 72.6          & 55.1          & 8.1           & 25.7          & 14.2          & 39.8          & 60.1          & 46.9          & 51.2          & 70.7          & 58.1          \cr
\multirow{15}{*}{}                           & \multicolumn{1}{|c|}{AttrGNN}          & 25.2          & 53.5          & 34.3          & 47.3          & 72.1          & 54.7          & 67.1          & 83.9          & 70.3          & 22.4          & 39.5          & 31.8          & 38.0          & 63.9          & 46.2          & 59.9          & 78.7          & 67.1          \cr
\multirow{15}{*}{}                          & \multicolumn{1}{|c|}{PoE}              & 12.6          & 25.1          & 17.0          & 46.4          & 65.8          & 53.3          & 66.6          & 82.0          & 72.1          & 11.3          & 22.9          & 15.4          & 34.7          & 53.6          & 41.4          & 57.3          & 74.6          & 63.5          \cr
\multirow{15}{*}{}                           & \multicolumn{1}{|c|}{PoE-rni}          & 23.2          & 39.0          & 28.3          & 38.0          & 55.7          & 44.2          & 50.2          & 64.1          & 55.8          & 25.0          & 49.5          & 33.4          & 41.1          & 66.9          & 49.8          & 49.2          & 70.5          & 57.2          \cr
\multirow{15}{*}{}                         & \multicolumn{1}{|c|}{MMEA}             & 26.5          & 54.1          & 35.7          & 41.7          & 70.3          & 51.2          & 59.0          & 86.9          & 68.5          & 23.4          & 48.0          & 31.7          & 40.3          & 64.5          & 48.6          & 59.8          & 83.9          & 68.2          \cr
\multirow{15}{*}{}                          & \multicolumn{1}{|c|}{EVA}              & 13.4          & 33.8          & 20.1          & 22.3          & 47.1          & 30.7          & 37.0          & 58.5          & 44.4          & 9.8           & 27.6          & 15.8          & 24.0          & 47.7          & 32.1          & 39.4          & 61.3          & 47.1          \cr
\multirow{15}{*}{}                           & \multicolumn{1}{|c|}{HEA}              & 12.7          & 36.9          & -             & 26.2          & 58.1          & -             & 41.7          & 78.6          & -             & 10.5          & 31.3          & -             & 26.5          & 58.1          & -             & 43.3          & 80.1          & -             \cr
\multirow{15}{*}{}                          & \multicolumn{1}{|c|}{ACK}         & 30.4          & 54.9          & 38.7          & 56.0          & 73.6          & 62.4          & 68.2          & 87.4          & 75.2          & 28.9          & 49.6          & 36.0          & 53.5          & 69.9          & 59.3          & 67.6          & 86.4          & 74.4          \cr
\multirow{15}{*}{}                          & \multicolumn{1}{|c|}{MCLEA}            & 29.5          & 58.2          & 39.3          & 55.5          & 78.4          & 63.7          & 73.0          & 88.3          & 78.4          & 25.4          & 48.4          & 33.2          & 50.1          & 70.5          & 57.4          & 66.7          & 82.4          & 72.2          \cr
\multirow{15}{*}{}                           & \multicolumn{1}{|c|}{MEAformer}        & \underline{40.2}    & \underline{70.3}    & \underline{50.4}    & \underline{61.5}    & \underline{83.6}    & \underline{69.5}    & \underline{76.2}    & \underline{91.0}    & \underline{81.8}    & \underline{32.4}    & \underline{59.0}    & \underline{41.4}    & \underline{55.8}    & \textbf{77.9} & \underline{63.8}    & \underline{70.3}    & \underline{86.9}    & \underline{76.5}    \cr \midrule
\multirow{2}{*}{\rotatebox{90}{Ours}}                         & \multicolumn{1}{|c|}{\textbf{DESAlign}}    & \textbf{49.7} & \textbf{75.0} & \textbf{58.6} & \textbf{65.6} & \textbf{85.3} & \textbf{72.8} & \textbf{80.5} & \textbf{92.6} & \textbf{85.0} & \textbf{41.0} & \textbf{66.0} & \textbf{49.5} & \textbf{57.3} & \underline{76.3}    & \textbf{64.2} & \textbf{72.8} & \textbf{87.7} & \textbf{78.2} \cr
\multirow{2}{*}{}
& \multicolumn{1}{|c|}{\textbf{Improv.}} & \textbf{9.5}  & \textbf{4.7}  & \textbf{8.2}  & \textbf{4.1}  & \textbf{1.7}  & \textbf{3.3}  & \textbf{4.3}  & \textbf{1.6}  & \textbf{3.2}  & \textbf{8.6}  & \textbf{7.0}  & \textbf{8.1}  & \textbf{1.5}  & -1.6          & \textbf{0.4}  & \textbf{2.5}  & \textbf{0.8}  & \textbf{1.7}  \cr 
\midrule
\multirow{4}{*}{\rotatebox{90}{Iterative}} & \multicolumn{1}{|c|}{EVA}              & 23.1          & 44.8          & 31.8          & 36.4          & 60.6          & 44.9          & 49.1          & 71.1          & 57.3          & 18.8          & 40.3          & 26.0          & 32.5          & 56.0          & 40.4          & 49.3          & 69.5          & 57.2          \\
\multirow{4}{*}{}                           & \multicolumn{1}{|c|}{MSNEA}            & 14.9          & 39.2          & 23.2          & 35.8          & 65.6          & 45.9          & 56.5          & 81.0          & 65.1          & 13.8          & 34.6          & 21.0          & 37.6          & 64.6          & 47.2          & 59.3          & 80.6          & 66.8          \\
\multirow{4}{*}{}                           & \multicolumn{1}{|c|}{MCLEA}            & 39.5          & 65.6          & 48.7          & 62.0          & 83.2          & 69.6          & 74.1          & 90.0          & 80.2          & 32.2          & 54.6          & 40.0          & 56.3          & 75.1          & 63.1          & 68.1          & 83.7          & 73.7          \\
\multirow{4}{*}{}                           & \multicolumn{1}{|c|}{MEAformer}        & \underline{57.1}    & \underline{81.2}    & \underline{65.6}    & \underline{68.5}    & \underline{87.4}    & \underline{75.2}    & \underline{78.4}    & \underline{92.1}    & \underline{83.5}    & \underline{44.4}    & \underline{69.2}    & \underline{52.9}    & \underline{60.8}    & \textbf{80.5} & \underline{67.8}    & \underline{71.8}    & \underline{87.9}    & \underline{77.9}    \\ \midrule 
\multirow{2}{*}{\rotatebox{90}{Ours}}                           & \multicolumn{1}{|c|}{\textbf{DESAlign}}    & \textbf{58.0} & \textbf{81.5} & \textbf{66.5} & \textbf{71.8} & \textbf{88.9} & \textbf{78.2} & \textbf{82.3} & \textbf{93.0} & \textbf{86.0} & \textbf{44.8} & \textbf{71.3} & \textbf{54.1} & \textbf{61.2} & \underline{79.9}    & \textbf{68.0} & \textbf{73.4} & \textbf{88.7} & \textbf{79.0} \\
\multirow{2}{*}{}                            & \multicolumn{1}{|c|}{\textbf{Improv.}} & \textbf{0.9}  & \textbf{0.3}  & \textbf{0.9}  & \textbf{3.3}  & \textbf{1.5}  & \textbf{3.0}  & \textbf{3.9}  & \textbf{0.9}  & \textbf{2.5}  & \textbf{0.4}  & \textbf{2.1}  & \textbf{1.2}  & \textbf{0.4}  & -0.6          & \textbf{0.2}  & \textbf{1.6}  & \textbf{0.8}  & \textbf{1.1}  \\ \bottomrule
\end{tabular}}
\end{table*}

\begin{table}[t!]
\caption{Main results of bilingual datasets.}
\vspace{-0.2cm}
\label{bilingual results}
\setlength{\tabcolsep}{3pt}
\centering
\resizebox{\linewidth}{!}{
\begin{tabular}{cccccccccc}
\toprule
Datasets                              & \multicolumn{3}{|c}{DBP15K$_{FR-EN}$}            & \multicolumn{3}{c}{DBP15K$_{JA-EN}$}           & \multicolumn{3}{c}{DBP15K$_{ZH-EN}$} \cr
\cmidrule(lr){1-1}\cmidrule(lr){2-4}\cmidrule(lr){5-7}\cmidrule(lr){8-10}
\multirow{2}{*}{Models}               & \multicolumn{9}{|c}{Non-iterative}                                                                                                             \\ \cmidrule{2-10} 
& \multicolumn{1}{|c}{H@1}           & H@10          & MRR           & H@1           & H@10          & MRR           & H@1           & H@10          & MRR           \\ 
\midrule
\multicolumn{1}{c|}{MUGCN}            & 49.5          & 87.0          & 62.1          & 50.1          & 85.7          & 62.1          & 49.4          & 84.4          & 66.1          \\
\multicolumn{1}{c|}{ALiNet}           & 55.2          & 85.2          & 65.7          & 54.9          & 83.1          & 64.5          & 53.9          & 82.6          & 62.8          \\
\multicolumn{1}{c|}{EVA}              & 68.3          & 92.3          & 76.7          & 67.3          & 90.8          & 75.7          & 68.0          & 91.0          & 76.2          \\
\multicolumn{1}{c|}{MSNEA}            & 54.3          & 80.1          & 63.0          & 53.5          & 77.5          & 61.7          & 60.1          & 83.0          & 68.4          \\
\multicolumn{1}{c|}{MCLEA}            & 71.1          & 90.9          & 78.2          & 71.5          & 90.9          & 78.5          & 71.5          & 92.3          & 78.8          \\
\multicolumn{1}{c|}{MEAformer}        & \underline{77.0}    & \underline{96.1}    & \underline{84.1}    & \underline{76.4}    & \underline{95.9}    & \underline{83.4}    & \underline{77.1}    & \underline{95.1}    & \underline{83.5}    \\ \midrule
\multicolumn{1}{c|}{\textbf{DESAlign}}    & \textbf{82.6} & \textbf{97.2} & \textbf{88.5} & \textbf{81.1} & \textbf{96.3} & \textbf{86.9} & \textbf{81.0} & \textbf{95.7} & \textbf{86.5} \\
\multicolumn{1}{c|}{\textbf{Improv.}} & \textbf{5.6}  & \textbf{1.1}  & \textbf{4.4}  & \textbf{4.7}  & \textbf{0.4}  & \textbf{3.5}  & \textbf{3.9}  & \textbf{0.6}  & \textbf{3.0}  \\ \midrule
\multicolumn{1}{c|}{Models} & \multicolumn{9}{c}{Iterative}\cr 
\midrule
\multicolumn{1}{c|}{BootEA}           & 65.3          & 87.4          & 73.1          & 62.2          & 85.4          & 70.1          & 62.9          & 84.7          & 70.3          \\
\multicolumn{1}{c|}{NAEA}             & 67.3          & 89.4          & 75.2          & 64.1          & 87.3          & 71.8          & 65.0          & 86.7          & 72.0          \\
\multicolumn{1}{c|}{EVA}              & 76.7          & 93.9          & 83.1          & 74.1          & 91.8          & 80.5          & 74.6          & 91.0          & 80.7          \\
\multicolumn{1}{c|}{MSNEA}            & 58.4          & 84.1          & 67.1          & 57.2          & 83.2          & 66.0          & 64.3          & 86.5          & 71.9          \\
\multicolumn{1}{c|}{MCLEA}            & 81.1          & 95.4          & 86.5          & 80.6          & 95.3          & 86.1          & 81.1          & 95.4          & 86.5          \\
\multicolumn{1}{c|}{MEAformer}        & \underline{84.5}    & \underline{97.6}    & \underline{89.4}    & \underline{84.2}    & \underline{97.4}    & \underline{89.2}    & \underline{84.7}    & \textbf{97.0}    & \underline{89.2}    \\ \midrule
\multicolumn{1}{c|}{\textbf{DESAlign}}    & \textbf{88.2} & \textbf{98.2} & \textbf{92.4} & \textbf{87.1} & \textbf{97.4} & \textbf{91.3} & \textbf{86.8} & \underline{96.9} & \textbf{90.9} \\
\multicolumn{1}{c|}{\textbf{Improv.}} & \textbf{3.7}  & \textbf{0.6}  & \textbf{3.0}  & \textbf{2.9}  & \textbf{0.0}  & \textbf{2.1}  & \textbf{2.1}  & -0.1 & \textbf{1.7}  \\ 
\bottomrule
\end{tabular}}
\end{table}

\begin{figure*}[t]
\centering
\includegraphics[width = \linewidth]{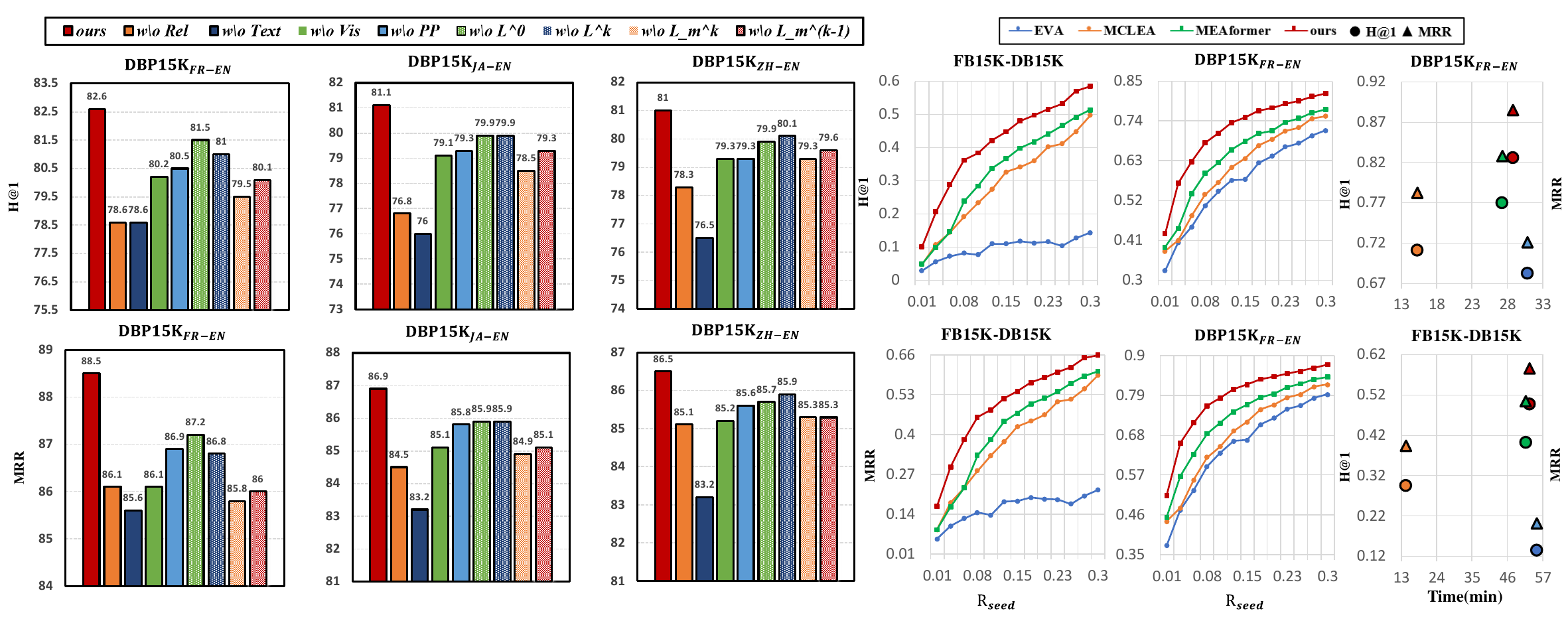}
\caption{The ablation study (left) and different supervised setting (right) for DESAlign.} 
\label{Ablation}
\end{figure*}

Our primary experiment is dedicated to assessing the performance of different methods under varying proportions of missing modalities, specifically $R_{img}$ (the ratio of images) and $R_{tex}$ (the ratio of text attributes). To regulate the level of semantic inconsistency, we select six representative proportions: $R_{img}, R_{tex} \in \{5\%, 20\%, 30\%, 40\%, 50\%, 60\%\}$ to simulate the degrees of missing modalities encountered in real-world scenarios. We thoroughly evaluate the robustness of the models and present the results while analyzing the performance from distinct perspectives.

\noindent
\textbf{(1) Robustness for Text Modality.} DESAlign demonstrates marked robustness in handling missing text modalities (Table \ref{Varying Text}), validating our theoretical propositions. Contrary to other methods that experience significant performance drops with incomplete text attributes—for instance, Meaformer's MRR decreases to 44.3\% (Table \ref{monolingual result} shows the original 50.4\%) on the FB-DB datasets—DESAlign consistently performs around 56.5\%, demonstrating enhanced stability and accuracy. It also consistently improves across different $R_{img}$ levels, with H@1 increases ranging from 9.8\% to 11.9\%.

\noindent
\textbf{(2) Robustness for Visual Modality.} As presented in Table \ref{Varying Image} for the DBP15K datasets, DESAlign demonstrates an average improvement on H@1 of 8.8\% ($R_{img} = 5\%$), 9.3\% ($R_{img} = 20\%$), 8.8\% ($R_{img} = 30\%$), 8.7\% ($R_{img} = 40\%$), 7.6\% ($R_{img} = 50\%$), and 6.3\% ($R_{img} = 60\%$). The most significant improvement is observed when $R_{img}$ is under 50\%. In contrast to the significant sensitivity of other methods, DESlign is more robust with varying ratios of missing images. Although the improvement is slightly lower compared to text modality, the overall advantage remains consistent, aligning with our motivations as concluded in Table \ref{Varying Text}.

\noindent
\textbf{(3) Robustness  Comparison.} Additionally, DESAlign optimally utilizes text and image modality, with alignment accuracy continuing to increase as the ratio of text attributes and images rises. This contrasts with existing models that exhibit performance oscillations or declines as the modality missing ratio increases. This adverse effect peaks within a specific ratio range but gradually recovers and gains benefits as the ratio rises to a certain level. Although counterintuitive, this observation is logical—introducing text attributes and images for half of the entities may introduce noise for the remaining half, necessitating a necessary trade-off.

\subsection{(Q2) Main Results on Real-world Graphs}

Our method undergoes rigorous evaluation on real-world multi-modal datasets, encompassing standard multi-modal DBP15K ($_{FR-EN}$, $_{JA-EN}$, $_{ZH-EN}$), FB15K-DB15K, and FB15K-YAGO15K. Monolingual and bilingual results from the experiments are presented in Table \ref{monolingual result} and Table \ref{bilingual results}, respectively. Beyond achieving superior performance on 18 benchmarks, including both basic and iterative strategies, we have the following other observations:

(1) DESAlign consistently surpasses the performance of state-of-the-art models across both bilingual and monolingual datasets. Particularly noteworthy are the remarkable improvements of 9.5\% and 8.6\% in the H@1 metric on the FB15K-DB15K ($R_{seed}=20\%$) and FB15K-YAGO15K ($R_{seed}=20\%$) monolingual datasets, respectively. Furthermore, DESAlign exhibits a 5.6\% improvement in the H@1 metric on the DBP15K$_{FR-EN}$ bilingual dataset. These results unveil that: (\romannumeral1) the semantic inconsistency issue is pervasive in real-world datasets; (\romannumeral2) the proposed semantic learning and semantic propagation strategies are highly effective for addressing this issue; and (\romannumeral3) the DESAlign model design demonstrates increased robustness in the face of severe semantic inconsistency issues. In terms of MRR, DESAlign showcases improvements ranging from 1.7\% to 8.2\% on monolingual datasets and 3.0\% to 4.4\% on bilingual datasets. These findings underscore DESAlign's versatility and reliability across diverse datasets, establishing it as a robust solution for MMEA.

(2)  The iterative strategy notably enhances multi-modal entity alignment by progressively incorporating high-quality alignments into the seed to refine the model. Specifically, DESAlign demonstrates significant improvements of 3.9\% accuracy on FB15K-DB15K ($R_{seed}=80\%$) datasets. Both H@1 and MRR metrics exhibit improvements across various datasets. Remarkably, even without the iterative strategy, DESAlign outperforms most iterative models on both bilingual and monolingual datasets. Moreover, as the ratio of seed alignments increases, the H@1 and MRR metrics of our method with an iterative strategy show gradual improvement. This is attributed to the iterative strategy leveraging previously predicted alignment results to enhance the training process. The increase in the number and quality of seed alignments contributes to greater improvements in the subsequent training processes. This exceptional performance underscores the effectiveness of our method in identifying high-quality aligned entity pairs for iterative training.

\subsection{(Q3, Q4) Analysis of DESAlign}

\noindent
\textbf{Ablation Study.}
We conduct an ablation study, evaluating various stripped-down versions of DESAlign presented in Figure \ref{Ablation} (left) to elucidate the H@1/MRR gains attributed to different components. We have the following key observations:

(1) Notably, the removal of content from any modality leads to a discernible performance decline, with the text attribute exhibiting the most significant impact. This emphasizes the importance of each modality in contributing to the overall alignment accuracy. Interestingly, the degradation effects of text attributes and relations under ablation appear similar, suggesting potential information complementarity between the text attribute and relation modalities. For instance, in DBP15K datasets, terms like \textit{clubs} and \textit{national team} may appear in both the $\mathcal{A}$ and $\mathcal{R}$ sets of the entity \textit{Mario Gomez}, indicating partial redundancy.

(2) Our analysis of the impact of each training objective, as detailed in Proposition \ref{MMSLearning}, reveals that the absence of any objective results in varying degrees of performance degradation. The objectives $\mathbf{X}^{(0)}$ and $\mathbf{X}^{(k-1)}$ play a role in constraining the bound of Dirichlet energy, with $\mathcal{L}{task}^{(0)}$ and $\mathcal{L}{m}^{(k-1)}$ enhancing the final semantic embedding quality. It is worth nothing that their influence is comparatively less significant than that of $\mathcal{L}{task}^{(k)}$ and $\mathcal{L}{m}^{(k)}$ in practice.

(3) While semantic learning mitigates the over-smoothing issue caused by semantic inconsistency in the training process, addressing the root cause of the lack of modal features for many entities is crucial. Semantic propagation proves effective in this regard, as evidenced by the substantial drop in accuracy when removing semantic propagation (w/o PP), impacting almost the same as removing an entire modality. This reveals that the semantic consistency issue truly has an impact on the accuracy of MMEA, and semantic propagation can precisely interpolate the missing modal feature.

\noindent
\textbf{Weakly Supervised Setting.}
To assess the method's stability in diverse environments, particularly in the weakly supervised setting, we evaluate DESAlign on two datasets with seed alignment ratios ($R_{seed}$) ranging from 0.01 to 0.30 in monolingual FB15K-DB15K and bilingual DBP15K$_{FR-EN}$. Figure \ref{Ablation} (right) illustrates a consistent performance gap between DESAlign and three other models as the ratio increases. This trend aligns with the conclusions drawn for three splits of FBDB15K and FBYG15K in normal supervised settings, demonstrating that DESAlign consistently outperforms the baseline with a significant performance gap, irrespective of the normal or weakly supervised setting. Notably, DESAlign achieves outstanding results, such as 42.8\% on Hits@1 in DBP15K$_{FR-EN}$ with only 1\% seed alignments and 28.8\% on Hits@1 in FB15K-DB15K with 5\% seed alignments, surpassing the prominent model MEAformer, which achieves 14.6\% in a normal supervised setting. In a weakly supervised environment with only 1\% alignment seeds, DESAlign even outperforms EVA under normal supervised settings, showcasing its potential for few-shot MMEA. This aspect will be further explored in our future works.

\subsection{Efficiency Analysis}
Our efficiency analysis of DESAlign alongside four MMEA methods reveals that DESAlign enhances performance with only a slight increase in computational time (Table \ref{Ablation}). This efficiency suggests DESAlign can easily replace existing models with minimal additional resource demands, marking efficiency improvement as an area for future exploration.

\textbf{Semantic Propagation}, a key component of DESAlign, offers a fast enhancement to alignments. The bulk of DESAlign's computational effort is dedicated to \textit{multi-modal semantic learning}, comparable to the MEAformer's semantic encoding process. Our \textit{semantic propagation (sp)} component requires merely 7 and 9 seconds on DBP15K and FB-DB datasets, respectively, with the main computational load stemming from sparse-to-dense matrix multiplication. This process has a complexity of $O(|\mathcal{E}|)$ (or $O(|\mathcal{E}|d)$ for vector features), which is linear in the number of entities, making it particularly efficient for large-scale graphs. Given its non-reliance on learning, \textit{sp} can be executed as a pre-processing step, optimizing its performance on CPU and accommodating large graphs beyond GPU memory limits. Furthermore, it seamlessly integrates as a plugin for enhancing other MMEA models.

\begin{figure}[t]
\centering
\includegraphics[width = \linewidth]{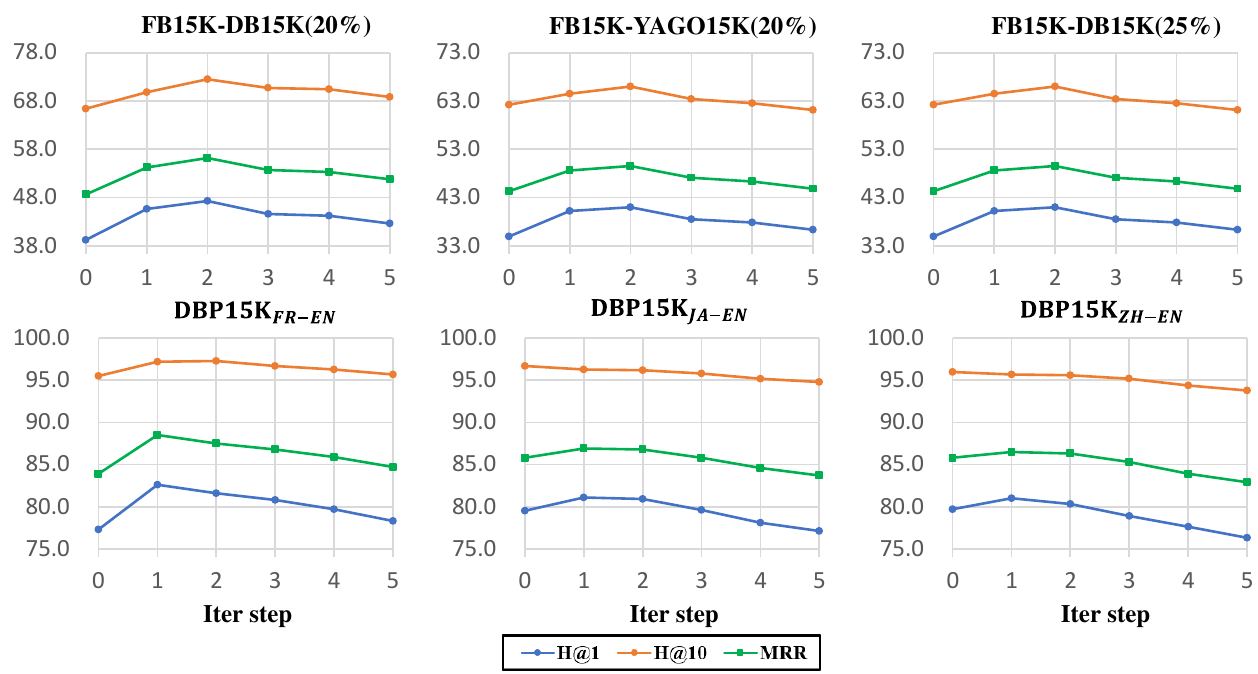}
\caption{The impact of number iterations in semantic propagation.} 
\label{propagationanalysis}
\end{figure}

\vspace{-0.2cm}
\subsection{Analysis of Semantic Propagation.}
We also study the impact of the number of iterations, $n_p$, we take in our semantic propagation. In Figure.\ref{propagationanalysis} we show the MMEA accuracy on both bilingual and monolingual datasets for different values of the number of iterations $n_p$. As we discussed in the Section.\ref{Semantic Propagation}, more propagations are equivalent to more gradient flow iterations, leading to more neighborhood semantic features used to estimations of the missing semantic feature. In practice, we make the semantic consistent features $\mathbf{x}_c$ join in the propagation process to simplify the application. It means too many processes will also import irrelevant semantic information from other features to the semantic consistent features, which misleads the alignments. We see that semantic propagation has different performances on various types of datasets. The number of propagation set as $n_p=2 $ gets the best performance for monolingual datasets, especially the $n_p=3$ iterations is the best setting for FB15K-DB15K($R_{seed} = 25\%$), and the $n_p=1$ will make best for bilingual DBP15K datasets. The difference in these effects is due to the structural and semantic differences inherent in bilingual KGs and monolingual KGs. Due to the greater heterogeneity in structure and semantics of bilingual KGs, collecting more semantic information from neighbors will also lead to the introduction of higher noise. But monolingual KGs utilize more semantic information with more propagation processes, while also ensuring that a smaller amount of noise is introduced. In addition, we see that when there are too many propagation processes, noise will have a higher impact, and the original semantic consistent features will also be misled to inconsistent.

\section{Related Work}
\label{relatedwork}
\textbf{Entity Alignment (EA)} approaches traditionally leverage structural features of KGs. Early logic-based methods, including logical reasoning and lexical matching, have been explored \cite{jimenez2011logmap, suchanek2011paris, sun2022revisiting, chen2017multilingual, zeng2020collective}. Recently, embedding-based methods gained popularity, with translation-based techniques like TransE \cite{bordes2013translating} and graph neural network (GNN)-based models \cite{tam2020entity}. The latter involves aggregating attributes \cite{cao2019multi, liu2020exploring}, relations \cite{chen2018co, wang2018cross}, and neighbor features \cite{liu2022guiding, liu2022selfkg, tang2023multi}. While attribute-enhanced techniques prove effective \cite{AttrE, COTSAE, liu2020exploring, tang2023cross}, existing methods ignore semantic inconsistency arising from diverse attributes.

\textbf{Multi-modal Knowledge Learning} focuses on leveraging diverse modalities for downstream applications, such as emotion recognition \cite{zhang2020emotion, jiang2021skeleton} and cross-modal retrieval \cite{zhen2019deep, chun2021probabilistic}. In the context of MMEA, existing studies \cite{UMAEA, meaformer, mclea} address modality discrepancies and missing features, particularly in the visual domain. However, they lack a holistic treatment of semantic consistency issues arising from diverse modalities. In contrast, we introduce a generalizable theoretical principle for guiding multi-modal knowledge graph learning and propose multi-modal semantic learning, offering optimal training under conditions of missing and noisy modalities while providing a comprehensive benchmark for evaluation.

\textbf{Multi-modal Entity Alignment} typically involves integrating visual and textual modalities to enhance KG-based entity alignment \cite{m6, b1}. Previous works, such as PoE \cite{liu2019mmkg}, represent entities as single vectors, concatenating features from multiple modalities. HEA \cite{m2} combines attribute and entity representations in hyperbolic space, utilizing aggregated embeddings for alignment predictions. Methods like MCLEA \cite{mclea} enhance intra-modal learning with contrastive methods, while MEAformer \cite{meaformer} improves modality fusion through hybrid frameworks. Despite their contributions, these approaches often overlook or inadequately address semantic inconsistencies arising from differences in the number or absence of specific modal attributes. Existing methods focus on creating high-quality distributions to compensate for missing attributes, introducing modal noise, and potentially losing valuable semantic information. In this work, we introduce semantic propagation, leveraging known features to restore missing semantic information across modalities.

\section{Conclusion}
\label{concluision}
This study presents DESAlign, a novel framework for multi-modal entity alignment that ensures semantic consistency across knowledge graphs. By proposing a Dirichlet energy-based principle, DESAlign successfully resolves semantic inconsistencies, combats over-smoothing when encoding multi-modal knowledge graphs, and interpolates missing semantics. Evaluated across a variety of public datasets under both standard and weakly supervised scenarios, DESAlign proves to be robust and superior, particularly in handling diverse modalities and extensive missing attributes. 

Moreover, our findings highlight potential areas for enhancing DESAlign's efficiency and the critical role of visual modality in enhancing alignment robustness.  Future work will explore optimizing these aspects, particularly focusing on developing efficient strategies and leveraging detailed visual content for improved multi-modal entity alignment.

\section*{Acknowledgment}

This work was supported in part by the National Natural Science Foundation of China under Grant 62201072, Grant U23B2001, Grant 62171057, Grant 62101064, Grant 62001054, and Grant 62071067; in part by the Ministry of Education and China Mobile Joint Fund under Grant MCM20200202 and Grant MCM20180101; in part by the BUPT-China Mobile Research Institute Joint Innovation Center.

\newpage
\bibliographystyle{IEEEtran}
\bibliography{reference.bib}

\end{document}